\documentclass[sigconf, 9cm]{acmart}
\usepackage{xspace}
\usepackage[capitalise,nameinlink,noabbrev]{cleveref}
\usepackage{enumitem}

\newtheorem{claim}{Claim}
\newtheorem{question}{Question}

\setcopyright{none}
\renewcommand\footnotetextcopyrightpermission[1]{} 

\newcommand{\defn}[1]{{\textit{\textbf{\boldmath #1}}}\xspace}
\renewcommand{\paragraph}[1]{\vspace{0.09in}\noindent{\bf \boldmath #1.}} 

\renewcommand{\epsilon}{\varepsilon}
\newcommand{\eps}{\varepsilon}
\newcommand{\bigO}{O}

\newcommand{\unk}{\texttt{UNK}}
\newcommand{\sss}{\texttt{SSS}}
\newcommand{\canc}{\texttt{CANC}}

\newcommand{\opt}{\texttt{OPT}}
\newcommand{\turtle}{\texttt{TURTLE}}
\newcommand{\alg}{\texttt{ALG}\xspace}

\newcommand{\rand}{\texttt{RAND}}
\newcommand{\bal}{\texttt{BAL}}
\newcommand{\equi}{\texttt{EQUI}}

\DeclareMathOperator{\trt}{\mathsf{TRT}}
\DeclareMathOperator{\mrt}{\mathsf{MRT}}
\DeclareMathOperator{\T}{\mathsf{T}}
\DeclareMathOperator{\C}{\mathsf{C}}
\DeclareMathOperator{\K}{\mathsf{K}}

\newcommand{\ser}{\circledcirc}
\newcommand{\pll}{\shortparallel}

\DeclareMathOperator{\work}{work}
\newcommand{\set}[1]{\left\{ #1\right\}}

\newcommand{\floor}[1]{\left\lfloor #1 \right\rfloor}
\newcommand{\ceil}[1]{\left\lceil #1 \right\rceil}
\newcommand{\paren}[1]{\left( #1 \right)}
\acmConference[SPAA '24]{}{June 17--21, 2024}{Nantes, France}
\copyrightyear{2024}
\acmYear{2024}
\setcopyright{licensedusgovmixed}\acmConference[SPAA '24]{Proceedings of the
36th ACM Symposium on Parallelism in Algorithms and Architectures}{June 17--21,
2024}{Nantes, France}
\acmBooktitle{Proceedings of the 36th ACM Symposium on Parallelism in Algorithms
and Architectures (SPAA '24), June 17--21, 2024, Nantes, France}
\acmDOI{10.1145/3626183.3659960}
\acmISBN{979-8-4007-0416-1/24/06}

\begin{document}

\title{Scheduling Jobs with Work-\emph{Inefficient} Parallel Solutions}

\author{William Kuszmaul}
\affiliation{
\institution{Harvard University}
\city{Cambridge}
\state{MA}
\country{USA}
}
\email{william.kuszmaul@gmail.com}
\authornote{
William Kuszmaul is funded by the Rabin Postdoctoral Fellowship in Theoretical Computer Science at Harvard University. Large parts of this research were completed while William was a PhD student at MIT, where he was funded by a Fannie and John Hertz Fellowship and an NSF GRFP Fellowship. William Kuszmaul was also partially sponsored by the United States Air Force Research Laboratory and the United States Air Force Artificial Intelligence Accelerator and was accomplished under Cooperative Agreement Number FA8750-19-2-1000. The views and conclusions contained in this document are those of the authors and should not be interpreted as representing the official policies, either expressed or implied, of the United States Air Force or the U.S. Government. The U.S. Government is authorized to reproduce and distribute reprints for Government purposes notwithstanding any copyright notation herein.
}
\author{Alek Westover}
\affiliation{
\institution{Massachusetts Institute of Technology}
\city{Cambridge}
\state{MA}
\country{USA}
}
\email{alekw@mit.edu}

\begin{abstract}
This paper introduces the \emph{serial-parallel decision problem}. Consider an online scheduler that receives a series of tasks, where each task has both a parallel and a serial implementation. The parallel implementation has the advantage that it can make progress concurrently on multiple processors, but the disadvantage that it is (potentially) work-inefficient. As tasks arrive, the scheduler must decide for each task which implementation to use.

We begin by studying \emph{total awake time}. We give a simple \emph{decide-on-arrival} scheduler that achieves a competitive ratio of $3$ for total awake time---this scheduler makes serial/parallel decisions immediately when jobs arrive. Our second result is an \emph{parallel-work-oblivious} scheduler that achieves a competitive ratio of $6$ for total awake time---this scheduler makes all of its decisions based only on the size of each serial job and without needing to know anything about the parallel implementations. Finally, we prove a lower bound showing that, if a scheduler wishes to achieve a competitive ratio of $O(1)$, it can have at most one of the two aforementioned properties (decide-on-arrival or parallel-work-oblivious). We also prove lower bounds of the form $1 + \Omega(1)$ on the optimal competitive ratio for any scheduler.

Next, we turn our attention to optimizing \emph{mean response time}. Here, we show that it is possible to achieve an $O(1)$ competitive ratio with $O(1)$ speed augmentation. This is the most technically involved of our results. We also prove that, in this setting, it is not possible for a parallel-work-oblivious scheduler to do well.

In addition to these results, we present tight bounds on the optimal competitive ratio if we allow for arrival dependencies between tasks (e.g., tasks are components of a single parallel program), and we give an in-depth discussion of the remaining open questions.

\end{abstract}





\maketitle

\section{Introduction}

There are many tasks $\tau$ for which the best parallel
algorithms are work inefficient. This can leave engineers with a
surprisingly subtle choice: either implement a serial version
$\tau^{\ser}$ of the task, which is work efficient but has no
parallelism, or implement a parallel version $\tau^{\pll}$ of the
task, which is work inefficient but has ample parallelism. The
serial version $\tau^{\ser}$ of the task takes some amount of
time $\sigma$ to execute on a single processor; the parallel
version $\tau^{\pll}$ takes time $\pi \ge \sigma$ to execute on a
single processor, but can be scaled to run on any number $k \le
p$ processors with a $k$-fold speedup. Which version of the task
should the engineer implement?

If the task is running in isolation on a $p$-processor system,
and assuming that $\pi / p \le \sigma$, then the answer is
trivial: use the parallel implementation $\tau^{\pll}$. But what
if the system is shared with many other tasks that are
arriving/completing over time? Intuitively, the engineer should
use $\tau^\pll$ if the system can afford to allocate at least
$\pi / \sigma$ processors to the task, and should use $\tau^\ser$
otherwise. But this choice is complicated by two factors, since
the number of processors that the system can afford to allocate
to $\tau$ may both (1) change over time as $\tau$ executes and
(2) depend on whether \emph{other} tasks $\tau'$ were executed
using \emph{their} serial or parallel implementations. The second
factor, in particular, creates complicated dependencies---the
right choice for one task depends on what choices have been and
will be made for others. 

In this paper, we propose an alternative perspective: What if the
engineer implements both a serial and parallel version of each
task, and then leaves it to the \emph{scheduler} to decide which
version to use? 

Formally, we define the \defn{serial-parallel decision problem}
as follows: a set $\mathcal{T} = \{\tau_1, \ldots, \tau_n\}$ of
tasks arrive over time. Each task $\tau_i \in \mathcal{T}$ arrives
at some start time $t_i$, and comes with two
implementations: a serial implementation $\tau_i^\ser$ with work
$\sigma_i$ and a parallel implementation $\tau_i^\pll$ with work
$\pi_i > \sigma_i$. In order for the scheduler to begin executing
a task $\tau_i$, it must choose (irrevocably) between which of
the two implementations to use. If the serial job $\tau_i^\ser$
is chosen, then the job can execute on up to one processor at a
time (the processor can change), and $\tau_i$ completes once 
$\sigma_i$ work has been performed on $\tau_i^\ser$. 
If the parallel job $\tau_i^\pll$ is
chosen, then the job can execute on up to $p$ processors over
time (now both the set of processors and the number of processors
can change), and $\tau_i$ completes once the total work performed on $\tau_i^\pll$ 
(by all processors) reaches $\pi_i$. 

Two natural objectives for the scheduler are
to minimize the \defn{mean
response time}($\mrt$), which is the average amount of time between when a task arrives and when it completes;
and the \defn{total awake time}, which is the total amount of time during which there is at least one job executing. We emphasize that, in both cases, our task is fundamentally to solve an online \emph{decision}
problem. If the scheduler was told for each task $\tau_i$ whether
to use $\tau_i^\ser$ or $\tau_i^\pll$, then the problem would
become straightforward. But actually making these
decisions is potentially difficult. 

\subsection*{Results}
In addition to formalizing the \emph{Serial-Parallel Scheduling Problem}, we give upper and lower bounds for how well online schedulers can perform on both awake time and average completion time.

\paragraph{Optimizing awake time} Our first result is a very simple scheduler that optimizes awake time with a competitive ratio of $3$. This is a \defn{decide-on-arrival scheduler}, meaning that it makes its serial/parallel decisions immediately when a job arrives. We also give lower bounds preventing a competitive ratio of $1 + \epsilon$: we show that any deterministic scheduler must incur competitive ratio at least $\phi - o(1) \approx 1.62$; and that any deterministic \emph{decide-on-arrival} scheduler must incur competitive ratio at least $2 - o(1)$; and that even \emph{randomized} schedulers must incur competitive ratios at least $(3 + \sqrt{3})/4 - o(1) \approx 1.18$. 

Our second result studies what we call \defn{parallel-work-oblivious schedulers}, which are schedulers that get to know the serial work of each job but \emph{not} the parallel work. We show that, even in this setting, it is possible to construct a 6-competitive scheduler for awake time. On the other hand, we show that there is a fundamental tension between \emph{decide-on-arrival} and \emph{parallel-work-oblivious} schedulers: any scheduler that achieves both properties has competitive ratio at least $\Omega(\sqrt{p})$.

\paragraph{Optimizing mean response time} Next, we turn our attention to mean response time (MRT). We construct an online scheduler that achieves a competitive ratio of $\bigO(1)$ for $\mrt$ using $\bigO(1)$ speed augmentation. This is our most technical result, and is achieved through three technical components. First, in \cref{sec:technicalLemma}, we prove two technical lemmas for comparing
the optimal schedule for a set of serial jobs to the optimal schedule
for perfectly scalable versions of the same jobs. Then, in
 \cref{sec:cancelcompetitive} we build on this to
construct a scheduler that is $\bigO(1)$ competitive \emph{if} it is permitted to sometimes
cancel parallel tasks and restart them as serial ones. This
cancellation `superpower' would seem to make the decision
problem significantly easier (as decisions are no longer
irrevocable). However, in \cref{sec:mainMRTdude}, we
show how to take our $\bigO(1)$-competitive scheduler (with cancellation) and transform it into a 
$\bigO(1)$-competitive scheduler (without cancellation). 

Our MRT scheduler is neither a decide-on-arrival scheduler nor an parallel-work-oblivious scheduler. We prove that, if a scheduler is decide-on-arrival and uses $O(1)$ speed augmentation, then it must incur competitive ratio $\Omega(p^{1/4})$. 

\paragraph{Extending to Tasks with Dependencies}
In \cref{sec:DTAP}, we extend our model to support arrival dependencies between tasks: Each task $\tau_i$ has a set $D_i \subseteq [n]$ of other tasks that must complete \emph{before} $\tau_i$ can arrive. Dependencies must be acyclic, but besides that, they can be arbitrary.

In this setting, it is helpful to think of the tasks as representing components of a \emph{single parallel program}. Each component has both a serial and parallel implementation, that the runtime scheduler can choose between. The goal is to minimize the completion time of the entire program---this corresponds to the awake time objective function.

In \cref{sec:DTAP}, we show that the optimal online competitive ratio for this problem (even for randomized schedulers) is $\Theta(\sqrt{p})$. The upper bound holds for any set of dependencies, and the lower bound holds even when the dependencies form a tree (this means that the lower bound applies, for example, even to fork-join parallel programs \cite{blumofe1995cilk}).


\paragraph{Open questions} Finally, in \cref{sec:conclusion}, we conclude with a discussion of open questions. 



\section{Related Work} 
\label{sec:relatedwork}

There is a vast literature on multiprocessor scheduling problems. For an excellent (but now somewhat dated) survey, see \cite{dutot_scheduling_2004}. Past work has often categorized sets of jobs $J=\set{j_1,\ldots,j_n}$ as being composed of jobs $j_i$ which are either rigid, moldable, or malleable. Both \defn{rigid} and \defn{moldable} jobs have the property that once a job $j_i$ begins on some number $p_{j_i}$ of processors, it must continue on that same set of $p_{j_i}$ processors without interruption until completion. Rigid and moldable jobs differ in that for rigid jobs the number $p_{j_i}$ of processors which task $j_i$ is to be run on is pre-specified, whereas for moldable jobs $p_{j_i}$ may be chosen by the scheduler. Finally, if the number (and set) of processors on which a job is executed is permitted to vary over time, then the job is said to be \defn{malleable}. In the contexts of moldable and malleable jobs, the jobs often come with \defn{speedup curves} determining how quickly the job can make progress on a given number of processors. If the speedup curve is proportional to the number of processors on which the job runs (as is the case for the parallel jobs associated with the tasks described in this paper) then the job is said to be \defn{perfectly scalable}. 

Much of the work in this area focuses on optimizing awake time, which as discussed earlier, is the total amount of time during which any jobs are present. Here, there has been a great deal of work on both offline schedulers \cite{Turek92, Mounie99, Tiwari94, Turek94_rigid, Turek94_moldable} and online schedulers \cite{Graham69, DuttonMao07, baker1983shelf, HurnikPaulus08, GuoKang10, Ye09, YeChen18}. 

For moldable jobs with arbitrary speedup curves, Ye, Chen, and Zhang
\cite{YeChen18} show an online competitive ratio of $O(1)$
for awake time. Of special interest to this paper would
be the speedup curve where job $j_i$ takes time $\sigma_i$ to
complete on $1$ processor and time $\pi_i / k$ to complete on $k
> 1$ processors. In this case, the scheduler's commitment to a
fixed number of processors would also implicitly represent a
commitment to running the job in serial or parallel. One difference between this and the problem studied here is that the scheduler (and the $\opt$ to which it is
compared) are \emph{non-preemptive}: they are required to execute the tasks on a fixed set of
processors (without interruption). Nonetheless, it is not too difficult to show that Ye, Chen, and Zhang's algorithm actually \emph{does} yield an $O(1)$ competitive awake-time algorithm for our problem---we emphasize, however, that this algorithm is neither decide-on-arrival nor parallel-work-oblivious, and would achieve a quite large constant competitive ratio. Interestingly, in the context of mean response time (MRT), we show in \cref{sec:mrtlower} that preemption is necessary: in the context
of our serial-parallel decision problem, any online scheduler
that rigidly assigns jobs to fixed numbers of processors will
necessarily incur a worst-case competitive ratio of $\omega(1)$
for $\mrt$ (even with $O(1)$ resource augmentation).

There is of course also a great deal of interest $\mrt$, i.e., the average amount of time between when tasks arrive
and when they are completed. Besides work on offline
approximation algorithms \cite{Turek94_moldable, Turek94_rigid},
most of the major successes in this area have been for malleable
jobs \cite{Edmonds00, Edmonds09, Ebrahimi18}. The seminal result
in this area is due to Edmonds \cite{Edmonds00}, who considered
malleable jobs with arbitrary sub-linear non-decreasing speedup
curves, and showed that the so-called $\equi$ scheduler, which
divides the processors evenly among all of the jobs present
(using time sharing if there are more than $p$ jobs), achieves a
competitive ratio of $O(1)$ with $2 + \epsilon$ speed
augmentation (subsequent work only requires speed augmentation $1
+ \epsilon$ with different schedulers \cite{Edmonds09,
Ebrahimi18}). A remarkable feature of the $\equi$ scheduler is
that it is \emph{oblivious} to the precise speedup curve of each
job. In \cref{sec:mrtlower}, we show that such a scheduler
is not possible in our setting---any oblivious online scheduler for $\mrt$ (or, even any parallel-work-oblivious scheduler)
must incur a worst-case competitive ratio of $\omega(1)$. 

The tasks studied in this paper do not fit neatly into the \\ rigid/moldable/malleable framework. They represent instead a direction that until now seems to have been unexplored: deciding for each task between the two extremes of (1) a fast algorithm with no parallelism and (2) a slower algorithm with ample parallelism and perfect scalability. Interestingly, several parts of the analysis end up making use of $\equi$ as an \emph{analytical tool}. Thus, for completeness, we restate Theorem 1.1 of \cite{Edmonds09} (with parameters $\beta = 1$ and $\epsilon = 1$) below:

\begin{theorem}[Theorem 1.1 of \cite{Edmonds09}]
$\equi$ with $3$ speed augmentation is $O(1)$ competitive with $\opt$ for MRT on any set $J$ of malleable jobs with arbitrary (nondecreasing sublinear) speedup curves.
\label{thm:jeff}
\end{theorem}

\section{Preliminaries} 
\paragraph{Problem Specification}
In this section we introduce our terminology and notation for
describing the problem. A \defn{task} is some computation that
must be performed. Tasks can be performed using a serial or
parallel \defn{job} implementing the task. In the
\defn{Serial-Parallel Scheduling Problem} a scheduler receives
tasks $\mathcal{T} = \set{\tau_1,\tau_2,\ldots, \tau_n}$ over
time, with $n$ unknown beforehand in the on-line case. Task
$\tau_i$ has an associated serial job $\tau_i^\ser$ with work
$\sigma_i$ and an associated parallel job $\tau_i^\pll$ with
work $\pi_i$. Finally, task $\tau_i$ arrives at time $t_i$ with
$t_1\leq t_2\leq \cdots \leq t_n$. Thus one should think of
$\tau_i$ as being determined by a triple $(\sigma_i, \pi_i,
t_i)$. 

The scheduler must decide whether to perform each task $\tau_i$
using the serial or parallel implementation. By default the scheduler need
not decide which implementation to run for a task exactly when
the task arrives, sometimes it may be beneficial to wait before
starting a task. We also consider the alternative model where the
scheduler must decide on arrival which implementation to use.

For convenience, we treat time steps as being small enough that
time is essentially continuous. At each time step, the scheduler
allocates its $p$ processors to the unfinished jobs present. In a
given time step multiple processors can be allocated to a
parallel job, but only a single processor can be allocated to
each serial job (and, of course, some jobs may be allocated $0$
processors). Sometimes it is convenient to treat a job as being
assigned to a fractional number of processors; this can be
accomplished by time-sharing the processor over multiple steps. A
serial job $\tau_i^{\ser}$ finishes once it has been allocated
$\sigma_i$ time (not necessarily contiguously) on a processor
(not necessarily the same one over time). A parallel job
$\tau_i^{\pll}$ finishes once it has been allocated $\pi_i$ total
time on processors---i.e., the integral over time of the number
of processors allocated to $\tau_i^{\pll}$ reaches $\pi_i$. 

We refer to a set of tasks $\mathcal{T} = \{\tau_1, \tau_2,
\ldots, \tau_n\}$ specifying an instance of the Serial-Parallel
Scheduling Problem as a \defn{task arrival process} or
\defn{TAP}. We use $[n]$ to denote $\{1,2,\ldots, n\}$, and
$\mathcal{T}_i^j$ to denote $\set{\tau_i, \ldots, \tau_j}$

\paragraph{Objective}
We say that a task is \defn{alive} if it has arrived but not yet
been completed. 
We consider two objective functions: 
\begin{itemize}
  \item minimize \defn{mean response time} ($\mrt$): the
    average amount of time that tasks are alive. If task $\tau_i$
    finishes at time $f_i$, then the MRT is $\frac{1}{n} \sum \left(f_i - t_i\right)$. 
It is equivalent, but generally more convenient, to work with a scaled version of $\mrt$ called
\defn{total response time} ($\trt$)---the sum of the amounts
of time that tasks are alive.
We denote the
$\trt$ of a scheduler $\alg$ on a TAP $\mathcal{T}$ by
$\trt_{\alg}^{\mathcal{T}}$.
  \item minimize \defn{awake time} ($\T$): the total amount of time that
there are alive tasks. If there are tasks alive on intervals
$[a_1,b_1]\sqcup [a_2,b_2] \sqcup \cdots$, then the awake time is
$\sum (b_i - a_i)$. We denote the awake time of a scheduler
$\alg$ on TAP $\mathcal{T}$ by $\T_\alg^\mathcal{T}$.
\end{itemize}

We measure a scheduler's performance by comparison to the optimal off-line
scheduler $\opt$, who can see the sequence of tasks in
advance. The \defn{competitive ratio} of a scheduler $\alg$ on
TAP $\mathcal{T}$ is the ratio of its performance to $\opt$'s,
e.g. $\trt_{\alg}^{\mathcal{T}} / \trt_\opt^\mathcal{T}$ for
$\trt$. 
More generally, we will say $\alg$ is \defn{$k$ competitive} if
the competitive ratio of $\alg$ is bounded by $k$ on all TAPs.
Sometimes we will also say $\alg$ is $k$ competitive with another
scheduler $\alg'$, meaning that the $\mrt$ (or awake time) of $\alg$ is never
more than a factor of $k$ larger than the $\mrt$ (or awake time,
respectively) of $\alg'$.

Finally, when comparing two schedulers $\alg_1$ and $\alg_2$ in
the context of $\mrt$ we will often assume \defn{$c$ speed
augmentation} for some $c \in \bigO(1)$. This means that $\alg_1$
gets to use processors that are $c$ times faster than those used
by $\alg_2$. Let $c\cdot \mathcal{T}$ denote the TAP
$\mathcal{T}$ but with every job's work multiplied by $c$;
similarly define $c\cdot J$ for a set of jobs $J$ as the jobs
from $J$ with work multiplied by $c$. Then, the statement
$\alg_1$ with $c$ speed augmentation is $\bigO(1)$ competitive
with $\alg_2$ on TAP $\mathcal{T}$ (or jobs $J$) can be written
as \[\trt_{\alg_1}^{\mathcal{T}} \leq \bigO(\trt_{\alg_2}^{c\cdot
\mathcal{T}}).\]

Our goal is to create a scheduler that is $\bigO(1)$ competitive
with $\opt$, potentially with use of $\bigO(1)$ speed
augmentation in the case of $\mrt$. 

\paragraph{Technical Details}
We emphasize that our focus is on schedulers that are allowed to
\defn{preempt} running tasks, i.e. pause running tasks and
continue later on a potentially different set of processors.
For $\mrt$, in particular, preemption is provably necessary---we show in  \cref{sec:mrtlower}
that a non-preemptive scheduler cannot be $\bigO(1)$-competitive
for $\mrt$. It is sometimes theoretically helpful to consider
schedulers that are allowed to \defn{cancel} running tasks, i.e.
stop a running task, erasing all progress made on the task, and
restart the task using a different implementation (e.g., a
parallel task can be cancelled and restarted as a serial task).
Our final schedulers \emph{will not} require cancelling. However,
in designing/analyzing our scheduler, it will be helpful to be
able to imagine what would have happened if cancelling were
possible. We say ``$\alg$, \defn{with use of cancelling}'' to 
denote that a scheduler $\alg$ has been augmented with the ability to cancel.

We consider only TAPs where the \defn{cost ratio}
$\pi_i/\sigma_i$ satisfies $\pi_i/\sigma_i \in [1,p]$ for all
tasks $\tau_i$. If $\pi_i/\sigma_i < 1$ then the scheduler
clearly should never run $\tau_i$ in serial so we can replace the
serial implementation with the parallel implementation to get
cost ratio $1$. Similarly, if $\pi_i/\sigma_i > p$ then the
scheduler should never run $\tau_i$ in parallel and we can
replace the parallel implementation with the serial
implementation to get cost ratio $p$. 

Throughout the paper we will assume $p\geq \Omega(1)$ is at least
a sufficiently large constant. We are more interested in the
asymptotic behavior of our schedulers as a function of $p$ than
the behavior on small values of $p$.

\section{A 3-Competitive Awake-Time Scheduler That Makes Decisions on Arrival}\label{sec:3comp}
The scheduler we describe in this section belongs to a simple
class of schedulers called \defn{decide-on-arrival} schedulers. 
Whenever a decide-on-arrival scheduler receives a task it must
immediately make an irrevocable decision about whether to run the
serial or parallel job associated with the task.
Our scheduler will run its chosen serial jobs via \defn{most-work-first}, i.e., run up to $p$ of the serial jobs with the most remaining work at each time step and run a parallel job on any remaining processors.  Clearly this is an optimal method for scheduling any given set of jobs.  Thus, in the decide-on-arrival model the Serial-Parallel Scheduling Problem is really not a scheduling problem but rather a decision problem: once the scheduler chooses which job to run for each task it is clear how to schedule the jobs.
We define $\opt$ to also be a decide-on-arrival scheduler.  This is  without loss of generality: $\opt$ does not benefit from delaying its decisions because $\opt$ has all the information in advance. Thinking of $\opt$ as a decide-on-arrival scheduler will simplify the terminology.
We now define specialized notation that is helpful in describing the decisions of decide-on-arrival schedulers:
\begin{definition}\label{def:decideonarrivaldefs}
  A scheduler $\alg$ is \defn{saturated} on time
  step $t$ if $\alg$ has no idle processors at time $t$.
  We say that $\alg$ is \defn{balanced} after the arrival of $n_0$ tasks if $\alg$ would be saturated immediately before completing these $n_0$ tasks, assuming no further tasks arrive. 
  In other words, being balanced means that \alg's current set of jobs
  could be distributed so that, assuming no additional tasks arrive, all processors
  would be in use at each time step until \alg has completed all the currently
  present jobs.
  If \alg is not balanced we say that $\alg$ is \defn{jagged}.



We say that \alg \defn{ incurs} work $W$ on a task $\tau$ if it takes a job requiring $W$ work to complete.
We say that \alg \defn{ incurs} work $W$ on a TAP $\mathcal{T}$ if \alg does $W$ total work to handle the tasks in $\mathcal{T}$.
\end{definition}
 We now present our simple $3$-competitive scheduler.
 \begin{theorem}
    \label{prop:3comp}
    There is a $3$-competitive decide-on-arrival scheduler for awake time.
\end{theorem}
\begin{proof}
    We propose the scheduler $\bal$ which is always
    balanced. Whenever a task $\tau$ arrives
    \begin{itemize}
        \item If taking $\tau^\ser$ would cause $\bal$ to
              become jagged $\bal$ takes $\tau^\pll$.
        \item Otherwise $\bal$ takes $\tau^\ser$.
    \end{itemize}
    To analyze $\bal$ it suffices to analyze TAPs where
    $\bal$ always has at least one uncompleted task present
    at any time from the start of time until the completion
    of the final task. Thus, the awake time of $\bal$ is the
    same as its completion time.
    $\opt$'s awake time may be less than $\opt$'s completion
    time; we call the intervals of time when $\opt$ has
    already completed all arrived tasks \defn{gaps}.
    We call the maximal intervals of time when $\opt$ has
    uncompleted tasks which have already arrived
    \defn{$\opt$-awake-intervals}.


    Our main technical lemma is the following:
    \begin{lemma}\label{lem:bolus_is_good}
        Fix some $\opt$-awake-interval $I$. 
        Let $\mathcal{T}(I)$ denote the set of tasks that arrive
        during $I$ and let $\mathcal{T}'$  denote the tasks that arrive before $I$.
        Assume that immediately before the start of $I$ $\bal$ has $B$ work present (and is of course balanced).
        Let $\C_\alg$ for $\alg\in \set{\bal,\opt}$ denote
        the time from the start of $I$ until when $\alg$
        would complete on the TAP $\mathcal{T}(I)\cup \mathcal{T}'$. Then
            $$\C_\bal \le B/p + 3\C_\opt.$$
    \end{lemma}
    \begin{proof}
        Let $\mathcal{T}(I, i)$ denote the first $i$ tasks in $\mathcal{T}(I)$. 
        Let $\K_\opt^i$ denote the work that $\opt$ incurs on  $\mathcal{T}(I,i)$.
        Let $\C_\alg^i$ for $\alg\in \set{\bal,\opt}$ denote completion time of $\alg$ on the TAP $\mathcal{T}'\cup \mathcal{T}(I,i)$ minus the start time of the $\opt$-awake-interval $I$.
        We will prove the lemma by induction on $i$, i.e., how many of the tasks in $T(I)$ to consider. One complication with the proof is that $\opt$'s schedule on $\mathcal{T}(I,i)$ can be very far from optimal; for instance, $\opt$ might schedule a very large task $\tau\in \mathcal{T}(I,i)$ in serial if it knows that future tasks in $\mathcal{T}(I)$ cause $\tau$ to not bottleneck completion time. Nonetheless with an appropriately constructed invariant we can control the evolution of the work profiles of  $\bal$ and $\opt$ as we increase $i$.
        
        In particular, we will inductively show the following claim:
        For each $i\in \set{0,1,\ldots, |\mathcal{T}(I)|}$ we have the invariant:
        \begin{equation} \label{eq:theinvariantbal}
            \C^{i}_\bal \le B/p + 2\C^i_\opt + \K^{i}_\opt / p.
        \end{equation}
        When $i=0$ we have $\C^0_\bal = B/p$ and the invariant is satisfied.
        Now we assume  the invariant is true for some 
        $i>0$ and consider how the addition of the 
        the $(i+1)$-th task  $\tau'$ impacts the invariant. We use  $\sigma',\pi'$ to denote the
        serial and parallel works of $\tau'$.
        If $\bal$ chooses an implementation of $\tau'$
        requiring $K\in\set{\sigma',\pi'}$ work
        and $\opt$ chooses an implementation of $\tau'$ requiring at least $K$
        work then
        $$\C^{i+1}_\bal - \C^i_\bal = K/p \le \K^{i+1}_\opt/p
            - \K^i_\opt/p$$
        and so the invariant for $i$ implies the invariant
        for $i+1$.

        The only remaining case is if $\opt$ runs $\tau'$ in
        serial while $\bal$ runs $\tau'$ in parallel.
        However, in this case $\tau'$ must have relatively
        large serial work.
        In particular, if we let $t_0'$ denote the time at
        the beginning of $I$ and $t_{i+1}'$ denote arrival time
        of $\tau'$ then
        $\sigma' > \C^i_{\bal}+t_0' - t_{i + 1}'$ or
        else $\bal$ would have chosen to run $\tau'$ in
        serial. Thus we have
        \begin{equation}\label{eq:optchosebigserialthing}
            \C^{i+1}_{\bal} = \C^i_{\bal} +
            \pi'/p \le \C^i_{\bal} +
            \sigma' \le t_{i+1}'-t_0'+ 2 \sigma'
            \le 2(t_{i+1}'-t_0'+\sigma').
        \end{equation}
        On the other hand, $\opt$ ran $\tau'$ in
        serial and does not have any gaps between $t_0',t_{i+1}'$ because these times occur during the same $\opt$-awake-interval $I$. Thus
        $$\C_\opt^{i+1} \ge t_{i+1}'-t_0'+\sigma'.$$
        Combined with \cref{eq:optchosebigserialthing} this gives
        $$\C^{i+1}_\bal \le 2\C_\opt^{i+1}, $$
        and so the invariant holds in this case as well.
        This completes the proof of the inductive claim.
        
        To conclude the lemma we use $i=|\mathcal{T}(I)|$ in the inductive claim, which gives:
        \begin{equation}\label{eq:bolusbolusbolus}
            \C_\bal \le B/p + 2\C_\opt + \K^{|\mathcal{T}(I)|}_\opt / p.
    \end{equation}
            Using the fact $\K_\opt^{|\mathcal{T(I)}|}/p \le \C_\opt$  in \cref{eq:bolusbolusbolus} completes the proof of the lemma.
    \end{proof}

    To finish our analysis of $\bal$ we inductively show that $\bal$ is
    $3$-competitive on any prefix of the $\opt$-awake-intervals. 
    As a base-case we can take the first $\opt$-awake-interval.
    Applying \cref{lem:bolus_is_good} with $B=0$ shows that 
    $\bal$ is $3$-competitive on the first $\opt$-awake-interval.
    The inductive step is:
    \begin{corollary} \label{cor:somemoreinduction}
        Let $\mathcal{T}^{(i)}$ denote the set of tasks that
        arrive during the first $i$ $\opt$-awake-intervals.
        Assume that $\bal$ is $3$-competitive with $\opt$ on
        $\mathcal{T}^{(i)}$. Then $\bal$ is $3$-competitive
        with $\opt$ on $\mathcal{T}^{(i+1)}$.
    \end{corollary}
    \begin{proof}
      Let $I$ denote the $(i+1)$-th $\opt$-awake-interval and let
      $\mathcal{T}(I)$ denote the set of tasks that arrive during
      $I$.
      Let $W$ be the work performed by $\bal$ before $I$ and let
      $B$ denote the work that $\bal$ has remaining immediately before the
      start of $I$.
      By assumption we have
      \begin{equation}\label{eq:beforethisinterval}
      (B+W)/p \le 3\T_{\opt}^{\mathcal{T}^{(i)}}.
      \end{equation}
      Let $\C_\alg$ for $\alg\in \set{\bal,\opt}$ denote the
      completion time of $\alg$ on the tasks $\mathcal{T}^{(i+1)}$ measured from the start of $I$; i.e., if $\alg$ completes $\mathcal{T}^{(i+1)}$ at time $t_\alg$ and $I$ starts at time $t_I$ then $\C_\alg=t_\alg-t_I$.
      By \cref{lem:bolus_is_good} we have
      \begin{equation} \label{eq:duringthisinterval}
      \C_{\bal} \le B/p + 3\C_{\opt}.
      \end{equation}
      Combining \cref{eq:beforethisinterval},
      \cref{eq:duringthisinterval} gives 
      \begin{align*}
      \T_\bal^{\mathcal{T}^{(i+1)}}  &= W/p + \C_\bal \\
      &\le  W/p + 3\C_{\opt} + B/p \\
      &\le 3(\T_{\opt}^{\mathcal{T}^{(i)}}+\T_\opt^{\mathcal{T}(I)}) \\
      &= 3\T^{\mathcal{T}^{(i+1)}}.
      \end{align*}
      
    \end{proof}

      Using \cref{cor:somemoreinduction} on all the
      $\opt$-awake-intervals shows that $\bal$ is $3$-competitive
      on $\mathcal{T}$.

\end{proof}

\section{A 6-Competitive Awake-Time Scheduler That is Parallel-Work Oblivious}\label{subsec:6comp}
Now we turn our attention to designing a \defn{parallel-work-oblivious} scheduler, which is a scheduler that is not allowed to see the parallel works of each task. 
We show that it is still possible to
construct an $\bigO(1)$-competitive scheduler for awake time in this model; this
contrasts with $\mrt$ where \cref{prop:oblivpllwork} asserts that knowledge of the parallel works is necessary to achieve a good competitive ratio for $\mrt$. In particular, we show:

\begin{theorem}\label{thm:3bill}
    There is a deterministic parallel-work-oblivious $6$-competitive scheduler for awake time.
\end{theorem}
We call our scheduler $\unk$. Whenever there are idle processors, $\unk$ takes any not-yet-started (but already arrived) task $\tau_i$ and:
\begin{itemize}
  \item If $\tau_i$ arrived more than $\sigma_i$ time 
  ago $\unk$ runs $\tau_i^\ser$.
  \item Otherwise $\unk$ runs $\tau_i^\pll$.
\end{itemize}
At each time step $\unk$ allocates a processor to each of the at-most-$p$ running serial jobs. There will be at most one parallel task running at a time. If there is a running parallel job $\unk$ allocates any processors not being used to run serial jobs to this single running parallel job. 
At any time some tasks may have arrived but not yet been started; that is, $\unk$ is not a decide-on-arrival scheduler.

Now we analyze $\unk$ using two lemmas.
As in \cref{sec:3comp} we say $\unk$ is saturated on a time step if all $p$ processors are in use, and unsaturated otherwise.
Also, as in the previous section, it suffices to analyze TAPs where $\unk$
always has at least one uncompleted task 
present at any time from the start of time
until the completion of the final task. We will make this assumption wlog for the rest of the section, meaning that the awake time of $\unk$ is equal to the completion time.

\begin{lemma}\label{lem:halftimeunsat}
$\unk$ is unsaturated at most $1/2$ of the time.
\end{lemma}
\begin{proof}
Let $W_1,W_2,\ldots$ denote the maximal intervals 
of time when $\unk$ is saturated; we call these \defn{saturated intervals}.
Similarly we refer to maximal intervals where $\unk$ is unsaturated as \defn{unsaturated intervals}. Observe that, whenever we are in an unsaturated interval, every task that has arrived so far either must have already completed or must be currently running in serial. That is, during an unsaturated interval there are never parallel jobs running and there are never jobs sitting around without being started.

For each saturated interval $W_i$
let $W_i'$ denote an interval of the same length as $W_i$
but shifted to start exactly at the end of interval $W_i$.
We claim that $\bigcup_i W_i'$ covers all unsaturated intervals.

Fix some unsaturated interval $[a,b]$.
Let $\tau_j$ be the serial task with the 
most remaining work present at time $a$.
Let $W_i$ be the saturated interval when $\tau_j$ was started.
We will show that $[a,b]\subseteq W_i'$.
Let $t$ denote the time when $\tau_j$ is started. 
Observe that $\unk$ must be saturated for all of $[t_j,t]$ or else $\tau_j$ would have been started in parallel. Thus, $[t_j, t]\subseteq W_i$.
Further, observe that $t-t_j\ge \sigma_j$, because 
$\unk$ must wait $\sigma_j$ time before starting $\tau_j$ in serial.
In particular this means that $|W_i| \ge \sigma_j$.

Next, we claim that at all times in $[t, b]$, $\unk$ allocates a processor to $\tau_j$.
Indeed, suppose that at some time step in $[t,b]$ $\tau_j$ was not allocated a processor. This would mean that there are $p$ other serial tasks with at least as much remaining work as $\tau_j$. But in this case $\unk$ would remain saturated until $\tau_j$ is completed, contradicting the fact that $\unk$ is not saturated during times $[a,b]$.
Thus, $b\le t+\sigma_j$, and so $[a,b]\subseteq [t, t+\sigma_j]$.
Finally, recall that $|W_i|\ge \sigma_i$ so $[a,b]\subseteq [t, t+|W_i|]$ and consequently $[a,b]\subseteq W_i'$.

Of course $\left|\bigcup_i W_i' \right| \le \sum_i |W_i|$. Thus $\unk$ is unsaturated at most $1/2$ of the time.

\end{proof}

\begin{lemma}
The amount of time that $\unk$ is saturated is at most $3\T_\opt.$
\end{lemma}
\begin{proof}
For each task $\tau$, let $A_\tau$ be the interval of time
between when $\tau$ arrives and when $\opt$ completes $\tau$.
Let $A=\bigcup_{\tau\in \mathcal{T}}A_\tau$ be the set of times when $\opt$ has uncompleted tasks, and let $B$ denote the set of times when $\opt$ has no uncompleted tasks present.
We divide tasks $\tau$ into four categories \footnote{The categories are not mutually exclusive. If a task $\tau$ falls in multiple categories we over-charge $\unk$ for $\tau$'s work.}:
\begin{enumerate}
  \item $\unk$ runs $\tau^\ser$.
  \item $\unk$ runs $\tau^\pll$ starting at some time after the end of $A_\tau$.
  \item $\unk$ runs $\tau^\pll$ entirely within time steps in $A$.
  \item $\unk$ starts $\tau^\pll$ during $A_\tau$, but part
    of $\tau$'s execution time by $\unk$ occurs during $B$.
\end{enumerate}

We now analyze the performance of $\unk$ on each category of
task.
\begin{claim}
  $\unk$'s total work on tasks of category (1) and (2) is at most $p\T_\opt$.
\end{claim}
\begin{proof}
For each task $\tau_i$ of category (1) $\opt$ must have incurred
at least $\sigma_i$ work; this is true of all tasks.
For each task $\tau_i$ of category (2) $\tau_i$ has not yet been
available for $\sigma_i$ time steps when $\unk$ starts $\tau_i$
in parallel. Thus, for $\opt$ to have already finished $\tau_i$
by the time that $\unk$ starts $\tau_i$ $\opt$ must have also
run $\tau_i$ in parallel. Thus, $\opt$ incurs work  $\pi_i$ for
task $\tau_i$. Therefore, $\opt$ incurs at least as much work as  $\unk$ on all tasks of categories (1) and (2). The amount of work that $\opt$ performs is at most $p\T_\opt$, which then bounds $\unk$'s work on tasks of category (1) and (2). \end{proof}

\begin{claim}
  $\unk$'s work on tasks of category (3) is at most $p\T_\opt$.
\end{claim}
\begin{proof}
  Tasks of category (3) all run during time steps in $A$, so the work spent on such tasks is at most
  $p\cdot |A| = p\T_\opt$.
\end{proof}

\begin{claim}
  $\unk$'s work on tasks of category (4) is at most $p\T_\opt$.
\end{claim}
\begin{proof}
Let $W_1^{\opt},W_2^{\opt}, \ldots$ denote the (maximal)
intervals of time when $\opt$ has uncompleted tasks. For each
$W_i^{\opt}$ there is at most one task $\tau_{k_i}$ that $\unk$ starts in parallel during $W_i^\opt$ whose execution time overlaps with $B$: this is because
$\unk$ only runs a single parallel task at a time. 
Let $K$ denote the set of category (4) tasks.
For each category (4) task $\tau_{k_i}\in K$, the corresponding $W_i^\opt$ must
have size at least $\pi_{k_i}/p$ because $\opt$ completes
task $\tau_{k_i}$ during $W_i^\opt$.
Thus, 
\begin{equation}\label{eq:sumsumsum}
\T_{\opt} = \sum_i |W_i^\opt| \ge \sum_{\tau_{k_i}\in K} \pi_{k_i}/p.
\end{equation}
The right hand side of \cref{eq:sumsumsum} is $\unk$'s
work on the category (4) tasks, giving the desired bound.
\end{proof}
Combining the previous three claims, the total work completed by
$\unk$ during saturated steps is at most $3p\T_{\opt}$. At each
saturated step $p$ units of work are performed. 
Thus, the total number of saturated time steps is at most $3\T_\opt$.
\end{proof}

Combined, the previous lemmas prove \cref{thm:3bill}.

\section{Minimizing $\mrt$}
\label{sec:meanresponsetime}
In this section we present our main result: a scheduler that, with $\bigO(1)$ speed augmentation, is $\bigO(1)$ competitive for $\mrt$ (or equivalently $\trt$). 

\subsection{Two Technical Lemmas}
\label{sec:technicalLemma}

In our analyses, it will be helpful to compare two settings: one in which a set of jobs $J$ must be executed with every job in serial, and the other in which the same set of jobs $J$ must be executed but where every job is perfectly scalable (i.e., $\pi_i=\sigma_i$). 

\begin{lemma}
  \label{lem:silly-serious}
  Let $J_{\ser}$ be a set of serial jobs with arbitrary arrival times. Let $J_{\pll}$ be jobs of
  the same work as jobs in $J_{\ser}$ but that are perfectly scalable.
  Then
 $$\trt_{\opt}^{J_{\ser}} \leq \bigO\paren{\trt_{\opt}^{O(1)\cdot J_{\pll}} +
 \sum_{j_i \in J_{\ser}} \work(j_i)}.$$  

\end{lemma}
\begin{proof} 
  We prove the lemma by constructing a scheduler $\sss$ that achieves
  mean response time at most 
  $$\bigO\paren{\trt_{\opt}^{12J_{\pll}} +
 \sum_{j_i \in 12J_{\ser}} \work(j_i)}.$$
  The \defn{Silly-Serious} scheduler
  $\sss$ operates in 2 modes: \defn{silly mode}, where there
  are less than $p$ unfinished jobs alive, and \defn{serious
  mode}, where there are at least $p$ unfinished jobs alive.
  Define \defn{serious intervals} and \defn{silly intervals} to
  be maximal contiguous sets of time where $\sss$ is in the
  respective modes. When discussing $\sss$, we will assume that
  it has access to $2p$ processors (rather than just $p$). This 
  can be simulated using a factor-of-$2$ speed augmentation and time sharing.

  During silly mode $\sss$ schedules each job on a single
  processor. As new jobs arrive, once the total number of jobs
  present reaches $p$, $\sss$ enters serious mode. As a boundary
  condition, we consider any jobs that arrive in that moment to
  have arrived \emph{during} the serious interval. We refer to the
  jobs that arrive during the serious interval as \defn{scary} (for
  this serious interval). 
  
  During a given serious interval, $\sss$ uses $2p$ total processors: it puts the at-most-$p$ non-scary jobs from the prior silly interval onto $p$ processors, and it schedules the scary jobs via $\equi$ on the other $p$ processors.

We remark that there may be fewer than $p$ scary jobs, in which
case $\equi$ may want to allocate multiple processors to a single
job. In this case $\sss$ does not actually schedule the (serial)
job on multiple processors but simply runs it on its own
processor. Because the jobs that $\sss$ is running are serial
only, we can think of them as having flat speedup curves, i.e.
they require the same amount of time to run regardless of how
many processors they are run on. So, when $\equi$ tries to run a
job on multiple processors, the progress is the same as if we
were to run it on a single processor. Thus we can think of $\sss$
as faithfully simulating $\equi$ on the scary jobs.

We bound the $\trt$ incurred by $\sss$ with $6$ speed
augmentation by partitioning the $\trt$ into three parts and
bounding each part.

We can bound the total $\trt$ incurred by $\sss$ during
silly intervals by $\sum_{j_i\in J_{\ser}} \work(j_i)$, because every job has a dedicated processor at all times during a silly interval. Thus we will focus the rest of the proof on serious intervals.

Now we fix a serious interval $I = [t_a,t_b]$ to analyze. Let
$X_{\ser}$ be the scary jobs for $I$, but with each job's work
truncated to be the amount of work that the job completes during
$I$. Similarly, let $Y_{\ser}$ be the non-scary jobs that run
during $I$, but with each of their works also reduced to be the
work completed by that job during $I$. 
We claim the following chain of inequalities holds for the jobs in the serious interval $I$:

\begin{align}
  \trt_{\sss}^{X_{\ser}\cup Y_{\ser}} & \leq \trt_{\equi}^{2(X_{\ser}\cup Y_{\ser})} \label{sssEQ2}\\
  &= \trt_{\equi}^{2(X_{\pll} \cup Y_{\pll})} \label{sssEQ3}\\
  &\leq \trt_{\opt}^{6(X_{\pll}\cup Y_{\pll})} \label{sssEQ4}\\
  &\leq \trt_{\opt}^{12X_{\pll}} + \sum_{j_i \in 12Y_{\pll}}\work(j_i).\label{sssEQ5}
\end{align}
Note that the first expression $\trt_{\sss}^{X_{\ser} \cup Y_{\ser}}$ represents the $\trt$ for $\sss$ during serious interval $I$. \\
\textbf{Inequality~\eqref{sssEQ2}:}
$\sss$'s treatment of $X_{\ser} \cup Y_{\ser}$ on $2p$ processors 
(which it is granted via factor-of-$2$ speed augmentation) is at
least as good as running $\equi$ on $X_{\ser} \cup Y_{\ser}$ with
$p$ processors, because $\sss$ runs $\equi$ on $X_{\ser}$
and then separately allocates one processor to each job in
$Y_{\ser}$.\\
\textbf{Inequality~\eqref{sssEQ3}:}
We denote by $X_{\pll}, Y_{\pll}$ perfectly scalable versions of the
jobs in $X_{\ser}, Y_{\ser}$. Since
there are at least $p$ jobs at all times during a serious
interval, $\equi$'s treatment of $2(X_{\ser}\cup Y_{\ser})$ and $\equi$'s
treatment of $2(X_{\pll}\cup Y_{\pll})$ are actually the same: it
equally partitions the processors amongst the available jobs,
potentially using time sharing; crucially $\equi$ never assigns
more than $1$ processor to any job because there are a
sufficiently large number of jobs. \\
\textbf{Inequality~\eqref{sssEQ4}:}
By Theorem~\ref{thm:jeff}
$\equi$ on $2(X_{\pll}\cup Y_{\pll})$ is $\bigO(1)$ competitive with
$\opt$ on $6(X_{\pll}\cup Y_{\pll})$. \\
\textbf{Inequality~\eqref{sssEQ5}:}
$\opt$ on $6(X_{\pll}\cup
Y_{\pll})$ using $p$ processors is at least as good as $\opt$ on $12(X_{\pll}\cup
Y_{\pll})$ using $2p$ processors. This, in turn, is at least good as the $\trt$ incurred by $\opt$ on $12X_{\pll}$ using $p$ processors plus the $\trt$ incurred by $\opt$ on $12Y_{\pll}$ using $p$ processors, i.e., $\trt_{\opt}^{6X_{\pll}} + \trt_{\opt}^{6Y_{\pll}}$. Finally, since $\trt_{\opt}^{12Y_{\pll}} \le \sum_{j_i\in 12Y_{\pll}}\work(j_i)$, we get \eqref{sssEQ5} as desired. \\
\textbf{Total $\trt$ incurred by serious intervals:} We can now bound the total $\trt$ incurred by serious intervals as follows. Let $I_1, I_2, \ldots$ be the serious intervals, and define $X^{(1)}_{\pll}, X^{(2)}_{\pll}, \ldots$ and $Y^{(1)}_{\pll}, Y^{(2)}_{\pll}, \ldots$ so that $X^{(k)}_{\pll}$ is $X_{\pll}$ for interval $I_k$ and $Y^{(k)}_{\pll}$ is $Y_{\pll}$ for interval $I_k$. The jobs in each $X^{(k)}_{\pll}$ and each $Y^{(k)}_{\pll}$ represent portions of jobs from $J_{\pll}$. Note that, although a given job from $J_{\pll}$ could appear in multiple $Y^{(k)}_{\pll}$'s, each job appears in at most one $X^{(k)}_{\pll}$ since each job can be scary in at most one serious interval.

By the inequalities above, we have that the total response time spent in serious intervals is at most
$$\sum_k \left(\trt_{\opt}^{12X^{(k)}_{\pll}} + \sum_{j_i\in 12Y^{(k)}_{\pll}}\work(j_i)\right).$$ 
Since each job in $J_{\pll}$ appears in at most one $X^{(k)}_{\pll}$, we have that
$$\sum_k \trt_{\opt}^{12X^{(k)}_{\pll}} \le \trt_{\opt}^{12J_{\pll}}.$$
Moreover, since the jobs in the $Y^{(k)}_{\pll}$'s each represent disjoint portions of the jobs in $J_{\pll}$, we have that
 $$\sum_k \sum_{j_i\in 12Y^{(k)}_{\pll}}\work(j_i) \le \sum_{j_i\in 12J_{\pll}}\work(j_i) \le \trt_{\opt}^{12J_{\pll}}.$$
 
Thus the total response time spent in serious intervals is at most $2 \trt_{\opt}^{12J_{\pll}}$, which completes the proof. 



\end{proof}

It will also be helpful to consider the setting in which we are comparing a set of perfectly scalable jobs to a set of 
serial jobs that arrive slightly later.

\begin{lemma}
   Let $J = \{j_1, \ldots, j_n\}$ be a set of perfectly scalable jobs, where job $j_i$ has work $2w_i$ and arrival time $t_i$. 
   Let $K = \{k_1, \ldots, k_n\}$ be a set of serial jobs where job $k_i$ has work at most $w_i$ and arrival time $t_i + w_i$. 
   Then,
   $$\trt^K_\opt \le O\left(\trt^{O(1) \cdot J}_{\opt} + \sum_i w_i\right).$$
   \label{lem:silly-serious2}
\end{lemma}
\begin{proof}
    Define $J'$  to be the set of jobs $J$, but where each job is serial rather than perfectly scalable. 
    Then 
    $$\trt^{K}_\opt \le \trt^{J'}_\opt,$$
    since when scheduling a serial job with $2w_i$ work, we would rather the job have $w_i$ less work than have the same job show up at time $w_i$ earlier.
    Finally, by \cref{lem:silly-serious}, 
    $$\trt^{J'}_\opt \le O\left(\trt^{12 \cdot J}_{\opt} + \sum_i w_i\right).$$
    This completes the proof.
\end{proof}

\subsection{A Cancelling Scheduler}
\label{sec:cancelcompetitive}

In this subsection we define and analyze a scheduler $\canc$ to prove 
\cref{thm:cancel}.
Note that $\canc$ uses cancelling, i.e. can kill tasks and restart them with a different
implementation; in Theorem~\ref{thm:nocancel} we remove the need for cancelling. 
\begin{theorem}
  \label{thm:cancel}
  There is an online scheduler which, using $\bigO(1)$ speed augmentation
  \textbf{and cancelling}, is $\bigO(1)$ competitive for $\mrt$.
\end{theorem}
We now describe the operation of $\canc$ on TAP $\mathcal{T}$.
$\canc$ starts by defining a set of ``\defn{relaxed jobs}'' $J'$ 
which incorporate the serial and parallel jobs from $\mathcal{T}$ into their speed-up curves; $\canc$ will simulate running jobs $J'$ as a subroutine to determine how to schedule $\mathcal{T}$.
In particular, for each task $\tau_i\in \mathcal{T}$ we form a relaxed job 
$j'_i\in J'$ with total work $2\sigma_i$ and the following speedup curve:
\begin{itemize}
  \item $j'_i$ receives no speedup on $x<\pi_i/\sigma_i$ processors.
  \item $j'_i$ receives speedup $x\cdot \sigma_i / \pi_i$ on $x\geq \pi_i/\sigma_i$ processors. 
\end{itemize}

When describing $\canc$ we will assume that it has access to $2p$ processors; this can be simulated using a factor-of-$2$ speed augmentation.
$\canc$ schedules $\mathcal{T}$ as follows:
  \begin{itemize}
    \item $\canc$ maintains a pool of $p$ processors for running parallel
      jobs and a pool of $p$ processors for running serial jobs.
    \item Initially tasks are placed in the parallel pool and will
      be run with their parallel implementation. 
    \item $\canc$ manages the parallel pool by simulating $\equi$
      on the relaxed jobs $j'_i$ and then actually running
      $\tau_i^{\pll}$ during the simulated execution slots for $j'_i$.
    \item Whenever a task $\tau_i$ been in the parallel pool for
    time at least $\sigma_i$, $\canc$ cancels task $\tau_i$ 
    (which was running as job $\tau_i^\pll$) and
    restarts $\tau_i$ as job $\tau_i^{\ser}$ in the serial pool.
    \item $\canc$ manages the serial pool with the $\equi$ strategy.
  \end{itemize}
  
 First we must establish that $\canc$ is a valid schedule, i.e. each task $\tau_i$ is completed by $\canc$. This is a concern because $\canc$ computes a schedule for the relaxed jobs, and assumes that the actual tasks are completed by running during the time slots of their corresponding relaxed job. 
 \begin{proposition}
     $\canc$ completes all tasks.
 \end{proposition}
\begin{proof}
   Tasks placed in the serial pool are clearly completed. 
   We proceed to argue that tasks never placed in the serial 
   pool are finished in the parallel
  pool. Consider some job $j_i'$ that finishes in the parallel
  pool, i.e. finishes in time less than $\sigma_i$; this
  corresponds to a task $\tau_i$ that is never placed in the
  serial pool because its corresponding relaxed job finishes in
  the parallel pool.
  We say that $j_i'$ executes in ``parallel mode'' when executing
  on at least $\pi_i/\sigma_i$ processors, and in ``serial mode''
  otherwise.

  For each $x \in [p]$, define $f_x$ to be the amount of time that $j_i'$ spends executing on (exactly) $x$ processors. Then the
  progress completed by job $j_i'$ is $$\sum_{x< \pi_i/\sigma_i}f_x +
  \sum_{x>\pi_i/\sigma_i}f_x \cdot x\cdot \sigma_i/\pi_i.$$ 
  Job $j_i'$ completes once it has made $2\sigma_i$ progress.
  At least half of the progress on $j_i'$ must have been made in
  parallel mode, because there is insufficient time to achieve
  $\sigma_i$ progress in serial mode.
  But this implies that
  $$\sum_{x>\pi_i/\sigma_i}f_x \cdot x \cdot \sigma_i/\pi_i\geq \sigma_i$$
  and thus that
  $$\sum_{x>\pi_i/\sigma_i}f_x  \cdot x \geq \pi_i.$$
  This implies that $\tau_i^{\pll}$, which also spends $f_x$ time on $x$ processors for each $x \in [p]$, successfully completes.

\end{proof}


The reason that we refer to relaxed jobs as ``relaxed'' is because there is a sense in which they are strictly easier to schedule than $\mathcal{T}$. We formalize this in the following lemma.

\begin{proposition}
    \label{prop:optTJlil}
$$\trt_{\opt}^{J'} \le \trt_\opt^{2\mathcal{T}} .$$
\end{proposition}
\begin{proof}
A schedule  for completing $2\mathcal{T}$ can be used to 
to perform $J'$ by running $j_i'$ in the time slot for $2\tau_i$. 
\end{proof}

We now bound the cost of $\canc$, thereby proving \cref{thm:cancel}.
\begin{proof}[Proof of \cref{thm:cancel}]
Let $J_{\ser}^1$ denote the serial jobs that end up in the serial
pool. Let $J_{\pll}^2$ denote the jobs in $J_{\ser}^1$ but modified to
be perfectly scalable. And let $J_{\ser}^2$ denote the jobs in
$J_{\ser}^1$ but with the arrival time of each job $j_i'$ arrival time delayed by $\sigma_i$ (i.e.,
delayed to be the time at which $j_i'$ is placed in the serial pool by $\canc$). 
$\canc$'s $\trt$ is bounded by:
\begin{align*}
\trt_{\canc}^{\mathcal{T}} \leq \trt_{\equi}^{J'}+\trt_{\equi}^{J_{\ser}^2}.
\end{align*}
By Theorem \ref{thm:jeff}, this is at most 
\begin{align*}
\trt_{\opt}^{3J'}+\trt_{\opt}^{3J_{\ser}^2}.
\end{align*}
By Lemma \ref{lem:silly-serious2}, 
$$\trt_{\opt}^{3J_{\ser}^2} \le O\left(\trt_{\opt}^{O(1) \cdot J_{\pll}^1} + \sum_{j \in J_{\ser}^2} \text{work}(j)\right).$$
Since $\trt_{\opt}^{O(1) \cdot J_{\pll}^1} \le \trt_{\opt}^{O(1) \cdot J'}$, it follows that
\begin{align*}
\trt_{\canc}^{\mathcal{T}} & \leq O\left(\trt_{\opt}^{O(1) \cdot J'}+ \sum_{j \in J_{\ser}^2} \text{work}(j)\right).
\end{align*}
In other words, defining $\mathcal{T}'$ to be the set of tasks in $\mathcal{T}$ that $\canc$ runs in serial mode, we have
$$\trt_{\canc}^{\mathcal{T}} \leq O\left(\trt_{\opt}^{O(1) \cdot J'}+ \sum_{\tau_i \in \mathcal{T}'} \sigma_i \right).$$
Notice, however, that by design, $\trt_{\canc}$ only runs a task $\tau_i$ in serial mode if, when we run $\equi$ on $J'$,
the job $j_i'$ incurs at least $\sigma_i$ response time. Thus 
$$\sum_{\tau_i \in \mathcal{T}'} \sigma_i \le \trt_\equi^{J'},$$
which by Theorem \ref{thm:jeff} implies that 
$$\sum_{\tau_i \in \mathcal{T}'} \sigma_i  \le O\left(\trt_{\opt}^{O(1) \cdot J'}\right).$$
Thus
\begin{align*}
\trt_{\canc}^{\mathcal{T}} & \leq O\left(\trt_{\opt}^{O(1) \cdot J'}\right),
\end{align*}
which by Proposition \ref{prop:optTJlil} completes the proof.
\end{proof}

\subsection{A Non-Cancelling Scheduler}
\label{sec:mainMRTdude}

Now we show how to convert $\canc$ from \cref{thm:cancel} to a non-cancelling scheduler. This is the most
technically difficult section of the paper.
\begin{theorem}
  \label{thm:nocancel} There is an online scheduler that, with
  $\bigO(1)$ speed augmentation and \textbf{without use of cancelling},
  is $\bigO(1)$ competitive for $\mrt$.
\end{theorem}
Up to a factor of 2 in speed augmentation, we can assume without loss of generality 
 that every $\sigma_i$ and $\pi_i$ is a power of two---to simplify our exposition, we shall make 
 this wlog assumption throughout the rest of the section.

We say that a task $\tau_k$ is of \defn{type $(2^j,2^i)$} if $\log \sigma_k = i$ and $\log \pi_k =i+j$. In other words, the job has parallelism $2^j$ and serial work $2^i$. 
We also partition the jobs into \defn{parallelism classes}, where the \defn{$2^j$ parallelism class} consists
 of tasks satisfying $\pi_k/\sigma_k = 2^j$. 

Our first lemma shows that we can modify the scheduler $\canc$ from \cref{thm:cancel} to 
obtain a ``just as good'' scheduler which only runs one task of each type in parallel at a time.
\begin{lemma}
   \label{lem:AB} 
   There exists an online scheduler $B$ that, \textbf{with cancelling} 
   and $\bigO(1)$ speed augmentation is $\bigO(1)$ competitive for $\mrt$.
   Furthermore, $B$ guarantees that at most one task 
   of each type is run via its parallel implementation at any time.
   Moreover, for any task $\tau_i$ that $B$ completes in parallel,
   $B$ is guaranteed to complete that task within time $\sigma_i$ of the task
   arriving.
\end{lemma}
\begin{proof}
  Recall that $\canc$ runs $\equi$ on tasks
  in the parallel pool until their serial time has elapsed. If a
  task $\tau_k$ is in the parallel pool for longer than $\sigma_k$, $\canc$
  cancels $\tau_k$ and restarts it in the serial pool via $\tau_k^\ser$.

  Our task is to construct $B$ so that 
  \begin{equation}
      \trt_{B}^{\mathcal{T}} \leq \trt_\canc^{\mathcal{T}}.\label{eq:abchain}
  \end{equation}

  $B$ runs $\canc$ on the parallel pool,
  except it concentrates all of $\canc$'s work on each task type
  into a single task of that type. That is, if $\canc$ allocates
  $k$ processors to tasks of type $(2^j,2^i)$ in the parallel pool on a certain time
  step, then $B$ will allocate $k$ processors to the current
  running task of type $(2^j,2^i)$ (if $B$ has a task of this type).
  The scheduler $B$ also copies $\canc$'s cancellation behavior as
  follows: when $\canc$ cancels a
  task of type $(2^j,2^i)$, $B$ attempts to cancel a task of the
  same type that is not currently running, and then restart that
  task in the serial pool; if there is only one task of the type $(2^j, 2^i)$
  running in $B$, then $B$
  cancels the running task and restarts it in the serial pool; and
  finally, if there are no tasks of this type in
  $B$, then $B$ does nothing and places a \emph{fake} task of
    type $(2^j, 2^i)$ in the serial pool. This behavior ensures that the
    serial pool for $B$ receives tasks of exactly the same types (and at exactly the same times) as the serial pool for $\canc$. 
  
  The point of this construction is that, at any given moment the
  number of tasks of each type that $B$ has either completed or
  evicted from the parallel pool is trivially guaranteed to be at least as large as that of
  $\canc$. 
  The serial pools of $\canc$ and $B$ are identical (with the
  fake tasks included). Hence the number of tasks alive for $B$
  at any given moment is at most as large as in $\canc$. So $B$
  achieves $\mrt$ at least as good as $\canc$ on $\mathcal{T}$.
\end{proof}

  Now we present a \emph{non-cancelling} scheduler $C$ that, with $\bigO(1)$
  speed augmentation, is $\bigO(1)$ competitive with $B$ for
  $\trt$.

  Up to a factor of 4 speed augmentation, we can assume that $C$ has $4p$ processors. We will make this assumption (without loss of generality) throughout the rest of the proof and keep the speed augmentation implicit. Thus, for the rest of the proof, we assume both that for every task $\tau \in\mathcal{T}$, $\sigma_i$ and $\pi_i$ are powers of two, and that $C$ is given $4p$ processors. 

   To clarify our exposition, when discussing $B$, we will make a distinction \defn{parallel work} (i.e., work on parallel jobs)
   and \defn{serial work} (i.e., work on serial jobs) performed by $B$. Note that the parallel work on a job in a time interval $[a, b]$ is defined as the integral over $[a, b]$ of the number of processors
   allocated to the job at each point in time.


   As we run the scheduler $C$, we will also simulate $B$ running $3\mathcal{T}$ with $p$ processors. 
   The scheduler $C$ will attempt to use $p$ of its processors to copy $B$'s behavior. Of course, as $B$ is 
   a cancelling scheduler, $C$ will not be able to precisely copy $B$. The challenge will be to 
   somehow achieve an $\mrt$ competitive with $B$'s $\mrt$, but without cancelling.

   Call a task in $C$ \defn{vested} if it has actually started
  executing in parallel in $C$. $C$ can copy $B$'s behavior except
  for when $B$ cancels a vested task (to be restarted in serial). 
  Whenever $B$ cancels a vested task $\tau$, the task $\tau$ enters 
  \defn{ballistic mode}. Whenever a task $\tau$ enters ballistic mode, we say that its 
  parallelism class enters \defn{emergency mode} (although the class may already
  be in emergency mode due to other tasks in the class already being in ballistic mode). 
  When a parallelism class is in emergency mode, all of the parallel work that $B$ performs
  is allocated by $C$ to the \emph{smallest} ballistic task in the parallelism class (we will see later that there are no ties here, but as we have not proven that yet, assume ties are broken arbitrarily).
  
Note that non-ballistic tasks lose work in $C$ compared to
$B$ when their parallelism class is in emergency mode (i.e., when a non-ballistic task $\tau$'s parallelism class is in emergency mode, it is possible that $C$ does work on $\tau$ while $B$ does not). 
If a non-ballistic task $\tau_i$
loses $q$ total parallel work to some ballistic task $\tau_k$, then we say that $\tau_k$ \defn{stole}
$q$ work from $\tau_i$.
We emphasize that this stolen work is not queued up to be done later by $\tau_i$, it is just lost. 
Thus it is possible for a task $\tau$ to finish running in parallel in $B$ \emph{without finishing in $C$}. 
If this happens, and $\tau$ is already vested in $C$, then the task $\tau$ also enters ballistic mode; 
otherwise, if $\tau$ is not yet vested in $C$, then the task $\tau$ enters what we call \defn{semi-ballistic mode}. Thus, a task $\tau$ enters ballistic mode if it is already vested and is then either cancelled
\emph{or completed} by $B$; and a task $\tau$ enters semi-ballistic mode if $B$ completes it in parallel, but if, at that point in time, $C$ has not even vested it.

The tasks in semi-ballistic mode are executed as \emph{serial jobs} on $p$ processors using $\equi$. Finally, the remaining $2p$ processors are allocated by $C$ as follows: $C$ allocates $p / 2^i$ processors to each parallelism class $2^i$. Whenever the parallelism class is in 
emergency mode, those processors are allocated as extra processors to the smallest ballistic task in the class.
This completes the description of $C$.

Now let us turn to the analysis of $C$. Let $\mathcal{T}_0$ denote the set of tasks that enter ballistic mode and $\mathcal{T}_1$ denote the set of tasks that enter semi-ballistic mode.

\begin{lemma}
    Each task $\tau_i \in \mathcal{T}_0$ spends at most $2\sigma_i$ time in ballistic mode.
    \label{lem:bal}
\end{lemma}
\begin{proof}
Whenever $\tau_i$ in parallelism class $2^j$ is actually \emph{executing} in ballistic mode (i.e., it is the smallest ballistic task from its parallelism class),
it is given at least $p / 2^j$ processors. Thus it spends at most $\pi_i / (p / 2^j) = \sigma_i$ time executing in ballistic mode. Additionally, $\tau_i$ may spend time in ballistic mode waiting on other (smaller) ballistic tasks to complete. 

Once a parallelism class enters emergency mode, it stops vesting new tasks. Let $t$ be the time at which $\tau_i$ entered ballistic mode, and let $t'$ be the most recent time $t' \le t$ at which the $2^j$ parallelism class entered emergency mode.
Any tasks in parallelism class $2^j$ that are ballistic at the same time as $\tau_i$ must have been running in $B$ at time $t'$. There can be at most one such task (including $\tau_i$) of each parallel power-of-two serial size $\sigma' \le \sigma_i$. Since each task in ballistic mode of some size $\sigma'$ spends at most $\sigma'$ time actually executing in ballistic mode, the total time that $\tau_i$ spends waiting on smaller ballistic tasks to finish is at most
$$\sum_{r = 1}^\infty \sigma_i / 2^r \le \sigma_i.$$
This complete the proof.
\end{proof}

\begin{lemma}
    The total response time incurred by the tasks $\mathcal{T}_1$ while in semi-ballistic mode in $C$ is at most
    $$O\left(\trt^{O(1) \cdot \mathcal{T}}_{\opt} + \sum_{\tau_i \in \mathcal{T}_0} \sigma_i\right).$$
    \label{lem:semibal}
\end{lemma}
\begin{proof}
    For each job $\tau_i \in \mathcal{T}_1$,  define $x_i$ to be a serial job
    of size $3\sigma_i$ whose arrival time is the time at which $\tau_i$ goes semi-ballistic in $C$;
    and define $y_i$ to be a perfectly scalable job
    of size $18\sigma_i$ whose arrival time is simply $t_i$.
    Let $X = \{x_i \mid \tau_i \in \mathcal{T}_1\}$ and let $Y = \{y_i \mid \tau_i \in \mathcal{T}_1\}$.
    The total response time incurred by the tasks $\mathcal{T}_1$ while in semi-ballistic mode in $C$ is at most
    $\trt^{X}_{\equi}$. By Theorem \ref{thm:jeff}, this is at most
    $\trt^{3X}_{\opt}$. By Lemma \ref{lem:silly-serious2}, this is at most
    $$O\left(\trt^{18 Y}_{\opt} + \sum_{\tau_i \in \mathcal{T}_0} \sigma_i\right).$$
    Since $\trt^{18 Y}_{\opt} \le \trt^{O(1) \cdot \mathcal{T}}_{\opt}$, the lemma follows. 



    
\end{proof}

The only way that a task can be present in $C$ but not in $B$ is if the task is either in ballistic or semi-ballistic mode. It follows by Lemmas \ref{lem:AB}, \ref{lem:bal}, and \ref{lem:semibal} that
\begin{equation}
\label{eq:CisballisticplusB}
\trt_C^{\mathcal{T}}\le O\left(\trt_{\opt}^{O(1) \cdot \mathcal{T}} + \sum_{\tau_i \in \mathcal{T}_0 \cup \mathcal{T}_1} \sigma_i\right).
\end{equation}

Thus, to complete the analysis of $C$, it suffices to bound 
$$\sum_{\tau_i \in \mathcal{T}_0 \cup \mathcal{T}_1} \sigma_i.$$
This is achieved through a charging argument in the following lemma.



\begin{lemma}
\label{lem:BC}
$$\sum_{\tau_i \in \mathcal{T}_0 \cup \mathcal{T}_1} \sigma_i \le \bigO(\trt_{\opt}^{3\mathcal{T}}).$$
\end{lemma}
\begin{proof}

If a task $\tau$ enters ballistic or semi-ballistic mode because $B$ placed $\tau$ into serial mode (i.e., $\tau$ was either
canceled by $B$ or was never even run in parallel mode by $B$), then call $\tau$ \defn{easy}.
By the analysis of $\canc$ the sum of the serial lengths of the
easy tasks is $\bigO(\trt_{\opt}^{O(1) \cdot \mathcal{T}})$. 

Call the other tasks that enter
ballistic or semi-ballistic mode \defn{hard}. The hard tasks are the ones that went ballistic 
or semi-ballistic only because
of other ballistic jobs stealing their parallel processing times---this causes the task to
complete in the parallel pool for $B$ without completing for $C$.

Each easy task $\tau_i$ is paid $2\sigma_i$ tokens upfront. Whenever any task
$\tau_i$ in $C$ has its parallel work (that $B$ wishes to perform on it) stolen
by a ballistic task $\tau_k$, the ballistic task $\tau_k$ pays tokens to the task $\tau_i$ proportionally to the work that is stolen: if both tasks are in the $2^j$ parallelism class, and task $\tau_k$ stole $q$ parallel processing time from task $\tau_i$ (here we are defining parallel processing time to be the integral over time of the number of processors that $\tau_k$ stole from $\tau_i$), then $\tau_k$ pays $\tau_i$ a total of $q / 2^j$ tokens. Finally, whenever a task $\tau_i$ enters either ballistic or semi-ballistic mode, the task pays
$\sigma_i$ tokens to the scheduling algorithm.

Since the sum of the serial lengths of the
easy tasks is $\bigO(\trt_{\opt}^{O(1) \cdot \mathcal{T}})$, the number of tokens paid to easy tasks upfront
is $\bigO(\trt_{\opt}^{O(1) \cdot \mathcal{T}})$. On the other hand, the number of tokens paid by 
ballistic and semi-ballistic tasks to the scheduler is 
$$\sum_{\tau_i \in \mathcal{T}_0 \cup \mathcal{T}_1} \sigma_i.$$
Thus, to complete the proof, it suffices to show that no task has a negative number of tokens when
it completes.

For each task $\tau_i$ that goes ballistic or semi-ballistic in some parallelism class $2^j$, the total number of tokens that it ever \emph{spends} is at most $\sigma_i$ (paid to the scheduler) plus $\pi_i / 2^j = \sigma_i$ (paid to other tasks that $\tau_i$ stole work from while being ballistic). Thus it suffices to show that
every task $\tau_i$ that goes ballistic or semi-ballistic \emph{earns} at least $2\sigma_i$ tokens during
its lifetime. 

If the task $\tau_i$ is easy, then it trivially receives $2\sigma_i$ tokens upfront. Otherwise, 
if a non-easy task $\tau_i$ in some parallelism class $2^j$ goes ballistic or semi-ballistic, 
then it must have had at least $2\pi_i$ parallel work stolen from it by ballistic tasks in its parallelism class
(because $B$ completed $3\pi_i$ parallel work on the task, but $C$ completed less than $\pi_i$). This
means that the task was paid at least $2\pi_i / 2^j = 2\sigma_i$ tokens, which completes the proof.

\end{proof}

Combining Lemma \ref{lem:BC} with \eqref{eq:CisballisticplusB}, we have completed the proof of Theorem \ref{thm:nocancel}.

\section{Generalization: TAPs with Dependencies}\label{sec:DTAP}
We now consider a generalization of the serial-parallel decision problem in which tasks can have dependencies---a given task $\tau_i$ will not arrive until all of the other tasks on which it depends are complete. For this section, we focus exclusively on optimizing awake time---note that, if the tasks correspond to components of a parallel program, the awake time corresponds to the completion time of the parallel program.

A \defn{DTAP} $\mathcal{D}$ (TAP with dependencies) is a set of tasks $\tau_i$ specified by $\sigma_i,\pi_i,t_i$ along with 
an associated set $D_i \subset [n]$
(potentially empty) of tasks that must be completed before task
$\tau_i$ can be started. 
That is, task $\tau_i$ becomes available only when the time $t$ satisfies $t>t_i$ and furthermore all tasks $\tau_j \in D_i$ have already been completed.
Of course, the dependency structure must form a DAG, or else it is impossible to run all tasks. We are interested in an online scheduler, which in this case means that the scheduler does not know anything about task $\tau_i$ until the task becomes available to run.

The following propositions establish a tight bound of $\Theta(\sqrt{p})$ on the optimal competitive ratio achievable by an online scheduler. The lower bound holds even for the case where the DTAP dependencies are required to form a tree.
\begin{proposition}
  For any (potentially randomized) online scheduler $\alg$, there exists a DTAP where $\alg$'s awake-time
  competitive ratio is $\Omega(\sqrt{p})$ with high probability in $p$.
\end{proposition}
\begin{proof}
  For convenience, we will discuss the DTAP as being randomly
  chosen from a set of DTAPs. Of course, if all schedulers $\alg$
  have competitive ratio $\Omega(\sqrt{p})$ with high probability
  for randomly chosen DTAPs from a class, then there exists some
  DTAP in the class for which a given $\alg$ is very likely to perform
  poorly on.

We consider a class of DTAPs which consist of $\floor{\sqrt{p}}$ \defn{levels}. 
Each level of these DTAPs consists of $\floor{\sqrt{p}}$ tasks with
serial work $1$ and parallel work $\sqrt{p}$. 
Of the $\floor{\sqrt{p}}$ tasks on each level exactly one randomly
chosen task spawns $\floor{\sqrt{p}}$ more tasks which form the next level. In particular, this single task is
the sole dependency for all tasks on the next level. All the
tasks in the DTAP have $t_i=0$, so each task arrives immediately once all the tasks it depends on are completed.

$\opt$, knowing the dependencies, could first run all the
spawning tasks via their parallel implementations to unlock all
tasks after time $\floor{\sqrt{p}} \sqrt{p} / p
\le 1$. Next, $\opt$ can schedule the remaining $\floor{p}^2-\floor{\sqrt{p}}$ serial tasks
via their serial implementations. Doing so $\opt$ achieves awake
time of at most $2$.

However, a scheduler $\alg$ that is unaware of the dependencies will
likely require much longer on this DTAP.
If $\alg$ is not willing to run more than $1/2$ of the tasks in a level
via parallel implementations then there is at least a $1/2$
chance that it requires time $1$ to pass the level due to running
the spawning task in serial. However, if
the $\alg$ is willing to run at least $1/2$ of the tasks in
parallel then in expectation it requires at least $1/4$ time to
uncover the spawning task. 
Either way, with constant probability $\alg$ spends $\Omega(1)$
time on each level. Thus, with high probability in $p$ the
scheduler requires $\Omega(\sqrt{p})$ total time to complete the
DTAP.
\end{proof}

\begin{proposition}
  There exists an online DTAP scheduler that is $\bigO(\sqrt{p})$ competitive for awake-time.
\end{proposition}
\begin{proof}
  Fix a DTAP $\mathcal{D}$ with $n$ tasks. It suffices to consider a DTAP where
  our scheduler always has at least one available uncompleted task, i.e., the case where its awake time and completion time are the same.

  We say that a task $\tau_i$ is \defn{fairly-parallel} if
  $\pi_i/p < \sigma_i/\sqrt{p}$, and \defn{not-very-parallel}
  otherwise. The $\turtle$ scheduler runs fairly-parallel tasks in parallel and
  not-very-parallel tasks in serial. $\turtle$ schedules the
  available tasks as follows:
  \begin{itemize}
    \item Whenever there is an available fairly-parallel task allocate all processors to a fairly-parallel task.
    \item If all available tasks are not-very-parallel, and there are $k$ such tasks, then allocate a processor to each of the $\min(p,k)$ present jobs with the largest remaining serial works.
  \end{itemize}

Now we analyze the performance of the $\turtle$ scheduler.
  First we consider the time that $\turtle$ spends running
  fairly-parallel tasks. Let $\mathcal{D}_\pll\subseteq
  \mathcal{D}$ be the fairly-parallel tasks. The time that
  $\turtle$ spends running tasks in $\mathcal{D}_\pll$ is
  $$\sum_{i \mid \tau_i \in \mathcal{D}_\pll} \pi_i / p 
  \le \sum_{i \mid \tau_i \in \mathcal{D}_\pll} \sigma_i / \sqrt{p} 
  \le \sqrt{p}\sum_{i\in [n]} \sigma_i / p
  \le \sqrt{p}\T_\opt^\mathcal{D}.$$
  Thus, in order to prove $\T_\turtle^\mathcal{D} \le
  \bigO(\sqrt{p}) \T_\opt^\mathcal{D}$ it remains to bound the time that  $\turtle$ spends executing not-very-parallel tasks.  Let $\mathcal{D}'$ be a new DTAP where we set the size of all fairly-parallel tasks to $0$.
  Observe that the time that $\turtle$ spends executing not-very-parallel tasks is the same in both $\mathcal{D},\mathcal{D}'$ because in $\mathcal{D}$ we always preempt not-very-parallel tasks if fairly-parallel tasks are available and run fairly-parallel tasks. Thus, it suffices to analyze $\T_\turtle^{\mathcal{D}'}$.
  
  Let $S$ be the amount of time that $\turtle$ on $\mathcal{D}'$ is saturated (i.e., has at least $p$ tasks) and $U$ be the amount of time that $\turtle$ is unsaturated. Let $\T_{\infty}^{\mathcal{D}'}$ be the time that it would take to perform $\mathcal{D}'$ by running each task in serial on its own processor (i.e., imagining that we had infinitely many processors). 
Observe that 
  $$U\le \T_\infty^{\mathcal{D}'}\le  \sqrt{p}\T_\opt^{\mathcal{D}'}$$
  where the second inequality is due to the fact that tasks in $\mathcal{D}'$ are not-very-parallel; in particular, if we ran each task in serial on its own $\sqrt{p}$-speed-augmented processor then the tasks would certainly complete no later than $\opt$.
    We also have
    $$S\le \sum_{i\in [n]} \sigma_i/p \le \T_{\opt}^{\mathcal{D}'}$$
    because 
  because total work is a lower bound on $\opt$'s
  awake time.
  Combining our bounds on $U,S$ we have we have
  $$ \T_\turtle^{\mathcal{D}'} \leq \bigO(\sqrt{p})\cdot \T_\opt^{\mathcal{D}'}.$$
\end{proof}

\section{Awake-Time Lower Bounds}\label{sec:awakelower}
In this section we present several lower bounds for awake time. 

\begin{proposition}
  \label{prop:evilLowerBound}
  No deterministic online scheduler can have competitive ratio
  smaller than $\phi-1/p$.\footnote{$\phi \approx 1.618$ denotes the golden
  ratio, i.e. the positive root of $x+1=x^2$.}
\end{proposition}
\begin{proof}
  Fix a deterministic scheduler $\alg$. Consider a TAP with
  $\sigma_1=\phi,\pi_1=p.$ If $\alg$ schedules $\tau_1$ in
  serial $\alg$ fails to be better than $\phi$ competitive in
  the case that no more tasks arrive. 
  If $\alg$ waits at least $1/\phi$ time before scheduling
  $\tau_1$ then $\alg$ also fails to be better than $\phi$
  competitive if no more tasks ever arrive.
  If instead $\alg$ schedules $\tau_1$ in parallel at some time
  $t_0<1/\phi$ and then $p-1$ un-parallelizable tasks 
  arrive right after $t_0$ with $\sigma_i = \phi-t_0$, then $\opt$ achieves awake time $\phi$ due to having chosen to run every task in serial
  while $\alg$'s awake time is at least 
  \begin{align*}
      t_0 + \frac{p+(p-1)(\phi-t_0)}{p} & \ge t_0 
 + \frac{p+p(\phi-t_0)}{p} - \phi/p \\
 & = 1 + \phi - \phi /p \\
 & = \phi^2 - \phi/p.
 \end{align*}
 So $\alg$'s competitive ratio is at least $\phi-1/p$.
\end{proof}

Now we consider the decide-on-arrival model. We can prove a stronger lower bound in the decide-on-arrival model than the \cref{prop:evilLowerBound}, and we conjecture that this separation is real.
\begin{proposition}
  No deterministic decide-on-arrival scheduler has competitive ratio better than ${2-\Omega\paren{\frac{\log p}{p}}}$.
  \label{prop:2lowerbounddecidearrive}
\end{proposition}
\begin{proof}
  Fix a deterministic scheduler $\alg$.
  Let $k=\floor{\log p}$.
  Consider the TAP $\mathcal{T}_1^p$ where $\sigma_i = 2^i, \pi_i = 2^{i-1}p$ for
  $i\in [k]$. In $\mathcal{T}_1^k$ tasks run twice as fast as their serial run time if they are fully parallelized. Furthermore
  task $\tau_i$ is twice as large as task $\tau_{i-1}$. 
  The arrival times for tasks in this TAP are separated by infinitesimal amounts of time. The final $p-k$ tasks of $\mathcal{T}_1^p$ are un-parallelizable tasks with $\sigma_i = 2^k$ which all arrive instantly after $\tau_k$.

First we consider the truncation of $\mathcal{T}_1^p$ to $\mathcal{T}_1^j$ for some $j\le k$.
$\opt$ will complete $\mathcal{T}_1^j$ by running tasks $\tau_1,\tau_2,\ldots, \tau_{j-1}$ in serial and task $\tau_j$ in parallel. Then, $\opt$'s awake time on $\mathcal{T}_1^j$ is
\begin{equation}
    \label{eq:12pprop2lowerbound}
\frac{1}{p} \left(2+4+\cdots+2^{j-1}+2^{j-1}p\right) \le (1+2/p)2^{j-1}.
\end{equation}
Now assume that $\alg$ runs $\tau_j$ in serial.
Then $\tau_j$'s awake time is at least $2^j$, which by \cref{eq:12pprop2lowerbound} means that $\alg$ has competitive ratio at least $2-\Omega(1/p)$ on $\mathcal{T}_1^j$.
That is, for $\alg$ to have any chance of achieving competitive ratio better than $2-\Omega(1/p)$, $\alg$ must schedule the first $k$ tasks in parallel.

Now we consider the performance of $\alg$ on $\mathcal{T}_1^p$ assuming that $\alg$ schedules the first $k$ tasks in parallel.
Here $\alg$ will have awake time at least
\begin{equation}
    \label{eq:awaketimealgpkpk}
\frac{1}{p} \left(1p+2p+\cdots+2^{k-1}p + (p-k)\cdot 2^k \right) \ge 2^k(2-k/p)-1.
\end{equation}
On the other hand, $\opt$ will run all $p$ tasks in serial and complete in time $2^k$. 
Thus, $\alg$'s awake time is at least $(2-\Omega(k/p))$-times larger than $\opt$'s on $\mathcal{T}_1^p$.
Plugging in our choice $k=\Theta(\log p)$ gives the desired bound on $\alg$'s competitive ratio.
\end{proof}

Although the lower bound of
\cref{prop:evilLowerBound} does not obviously translate
to randomized schedulers we can still show a $1+\Omega(1)$ lower bound for such schedulers as follows. 
\begin{proposition}\label{prop:randlb}
  Let $\rand$ be a randomized scheduler. 
  There exists a TAP on which $\rand$ has competitive ratio 
  at least $\frac{3+\sqrt{3}}{4} - \Theta(1/p)$ with probability arbitrarily close to $1$.

\end{proposition}
\begin{proof}
  In $\mathcal{T}_A$ a single task
   with $\sigma_1= \sqrt{3}+1$, $\pi_1 =2p$ arrives at time $0$.
TAP $\mathcal{T}_B$ starts the same as $\mathcal{T}_A$, but in
$\mathcal{T}_B$ $p-1$ additional maximally
un-parallelizable tasks arrive at time $1$ with serial work
$\sqrt{3}.$
A simple calculation shows that no deterministic algorithm can achieve competitive ratio better than $\frac{1+\sqrt{3}}{2}-\Theta(1/p)$ on both $\mathcal{T}_A$ and $\mathcal{T}_B$.
Furthermore $\rand$ cannot have better than a $1/2$ chance of
  performing well on both $\mathcal{T}_A$ and $\mathcal{T}_B$.
  Thus, $\rand$'s expected competitive ratio on at least one of
  these TAPs is at least $\frac{1+\frac{1+\sqrt{3}}{2}}{2} -
  \Theta(1/p) = \frac{3 + \sqrt{3}}{4} - \Theta(1/p).$

Let $s_1,s_2,\ldots,s_n$ be a random string of $A$'s and $B$'s.
Now we form a TAP $\mathcal{T}$ by placing $\mathcal{T}_{s_i}$ at time $10i$.
  On $\mathcal{T}$ $\rand$ handles each choice between $\mathcal{T}_A, \mathcal{T}_B$ correctly with probability at most $1/2$. Thus, for any $\eps>0$ if we make the sequence sufficiently long (i.e., take $n$ large enough) then $\rand$ has arbitrarily low probability of handling more than a $(1/2+\eps)$-fraction of the $\mathcal{T}_A, \mathcal{T}_B$ choices correctly. Thus, 
    $\rand$ has competitive ratio at least at least $\frac{3+\sqrt{3}}{4} - \Theta(1/p)$ with probability arbitrarily close to $1$.
\end{proof}

Now we consider parallel-work-oblivious schedulers. We show that, even if all tasks arrive at a single time, there is a lower bound of $2 - o(1)$ on the optimal competitive ratio.
\begin{proposition}
    There is no deterministic parallel-work-oblivious scheduler that achieves competitive
    ratio better than $2-\Omega(1/p)$ for awake time, even in the single-arrival-time setting.
\end{proposition}
\begin{proof}

    Fix a deterministic parallel-work-oblivious scheduler $\alg$.
    Consider a TAP with two tasks $\tau_1,\tau_2$ both of serial size $1$.
    Technically there are many different ways that $\alg$ can
    handle $\tau_1,\tau_2$. We will reduce the space of possible
    strategies that $\alg$ can employ by showing that certain
    strategies are dominant over other strategies. Combined with
    some case analysis this will allow us to show that there
    must be some TAP on which $\alg$ performs poorly.

    Without loss of generality $\alg$ instantly starts (at least)
    one of the tasks.
    If $\alg$ starts a serial job at time $0$ then $\alg$ has
    competitive ratio at least $\Omega(p)$ on the TAP where $\tau_1,\tau_2$ are perfectly scalable.
    Thus, it suffices to consider the case where $\alg$ starts by running one of the tasks in parallel; call this instantly
    started task $\tau_1$.

    Without loss of generality $\alg$ runs $\tau_1$ on all processors at each time step until it starts $\tau_2$.
    If $\tau_1$ completes before $\alg$ starts $\tau_2$ then without loss of generality $\alg$ instantly starts $\tau_2$ in parallel.

    Now we can describe $\alg$ as follows: at each time $t\in[0,1]$
    until $\tau_1$ finishes $\alg$ can decide whether to start $\tau_2$ at time $t$ and which implementation of $\tau_2$ to use when starting it.
    Let $x_\alg\in [0,1]$ be the earliest time when $\alg$ is willing to start $\tau_2$ even if $\tau_1$ is not yet completed by this time.
    If $\alg$ chooses to schedule $\tau_2$ in parallel at time $x_\alg$ then there is a TAP on which $\alg$ has competitive ratio at least $2$: namely, the TAP where $\tau_1,\tau_2$ are both completely un-parallelizable.
    Thus, it suffices to consider the case where $\alg$ would choose to run $\tau_2$ in serial at time $x_\alg$ assuming that $\tau_1$ has not yet finished by time $x_\alg$.
    
   By now we have substantially simplified the description
    of $\alg$. In particular, we have shown that $\alg$ is
    completely parameterized by a single value $x_\alg \in [0,1]$.
    Given $x_\alg$, we have reduced to the case where $\alg$'s strategy is
    \begin{enumerate}
        \item Run $\tau_1^\pll$ from the start.
        \item If $\tau_1$ finishes before time $x_\alg$ start
              $\tau_2^\pll$ immediately once $\tau_1$ finishes.
        \item Otherwise, start $\tau_2^\ser$ at time $x_\alg$.
    \end{enumerate}

    To conclude we consider two cases on the value of $x$.
    If $x_\alg<1-1/p$ then $\alg$ performs poorly on a TAP where
    $\pi_1/p = x+1/p$ and $\pi_2/p = 1/p$.
    Indeed, for this TAP $\opt$ has awake time $x+2/p$ whereas $\alg$
    has awake-time at least $x+1$. Thus that $\alg$'s competitive ratio here is at least 
    $$\frac{x+1}{1+2/p}\ge 2-\Omega(1/p)$$

    If instead $x_\alg\geq 1-1/p$ then $\alg$ performs poorly on the TAP where $\tau_1,\tau_2$ are both completely un-parallelizable. In this
    case $\opt$ achieves awake time $1$ by running the tasks in serial from the start whereas $\alg$ has awake time at least $2-1/p$.

    Thus, regardless of $x_\alg$, $\alg$ has competitive ratio at least $2-\Omega(1/p)$ on some TAP.
\end{proof}

Finally, we show that if we simultaneously restrict to the decide-on-arrival model and the parallel-work-oblivious model then the
scheduler cannot perform well. 
\begin{proposition}
  Any deterministic scheduler $\alg$ that is both decide-on-arrival and parallel-work-oblivious must have competitive ratio $\Omega(\sqrt{p})$.
\end{proposition}
\begin{proof}
  Consider the following two TAPs: 
  \begin{enumerate}[label=(\alph*)]
    \item $\ceil{\sqrt{p}}$ identical scalable tasks
      arrive at the start.
    \item  $\ceil{\sqrt{p}}$ identical un-scalable tasks arrive at the
      start.
  \end{enumerate}
  To $\alg$ these two TAPs look identical. However, if $\alg$
  decides to run any task in serial then it's competitive ratio on (a) is $\Omega(\sqrt{p})$.
  Otherwise, if $\alg$ decides to run all tasks in parallel then its competitive ratio on (b) is $\Omega(\sqrt{p})$.
  Note that this simple decomposition into two cases is made possible by the assumption that $\alg$ is a decide-on-arrival scheduler.
\end{proof}

\section{$\mrt$ Lower Bounds}
\label{sec:mrtlower}
In this section we prove two lower bounds on schedulers for $\mrt$. 
These bounds show that any $\bigO(1)$ competitive scheduler must (1) make
decisions at least partially based on the values of $\{\pi_i\}$;
and (2) must be willing to vary the number of processors assigned to a given job over time.
Thus these properties cannot be relaxed in the schedulers presented in the previous sections.

We remark that our first lower bound applies not just to schedulers that achieve worst-case competitive ratios,
but also to schedulers that use randomization in order to achieve a bounded \emph{expected} competitive ratio.



\begin{proposition}
  \label{prop:oblivpllwork}
  Fix an online scheduler $\alg$ that is oblivious to the
  parallel works of tasks, and fix some $c\in \Theta(1)$. There
  exists a TAP on which the expected competitive ratio of $\alg$ with $c$
  speed augmentation is at least $\Omega(p^{\frac{1}{4}})$ for $\mrt$.
\end{proposition}
\begin{proof}
  Consider a TAP with $\sqrt{p} + p^{1/4}$ tasks, all with serial work $1$. Suppose that
  $\sqrt{p}$ of the tasks have parallel work $1$, and that (a random subset of) $p^{1/4}$
  of the tasks have parallel work $p$. Call these the \defn{cheap} and \defn{expensive}
  tasks, respectively. All of the tasks arrive at time $0$.

  If the cheap tasks are run in parallel, and then the expensive tasks are run in serial, then
  the $\trt$ will be $O(\sqrt{p} \cdot \frac{1}{\sqrt{p}} + p^{1/4} \cdot 1) = p^{1/4}$.

  Now suppose for contradiction that $\alg$ also achieves $O(p^{1/4})$ $\trt$ using $O(1)$ speed augmentation.
  Let $\delta$ be the expected fraction of the tasks that $\alg$ runs in serial. The expected
  number of cheap jobs that are executed in serial is $\delta \sqrt{p}$. Thus, in order for
  $\alg$ to be $O(1)$-competitive (even with $O(1)$ speed augmentation), we would need $\delta \le O(p^{-1/4})$. 
  But this means that the expected number of expensive jobs that are executed in parallel is at least
  $(1 - \delta) p^{1/4} \ge \Omega(p^{1/4})$. If $k$ expensive jobs are executed in parallel, their
  $\trt$ will be at least $\Omega(k^2)$. By Jensen's inequality, the expected $\trt$ of expensive jobs executed
  in parallel is therefore at least $\Omega((p^{1/4})^2) = \Omega(\sqrt{p})$. This contradicts the assumption
  that $\alg$ achieves $\trt$ $O(p^{1/4})$. 
\end{proof}

\begin{proposition}
  Consider an online scheduler $\alg$ that is non-preemptive (i.e., the number of processors that it assigns to each task is fixed for the full duration of time that the task executes). Fix some
  $c\in \Theta(1)$, and any $R \in \mathbb{N}$. There exists a TAP
  $\mathcal{T}$ on which the worst-case competitive ratio of $\alg$ with $c$ speed augmentation is $\Omega(R)$ for $\mrt$.
  That is, $\alg$  with performs arbitrarily worse than $\opt$.
\end{proposition}
\begin{proof}
As we are interested  in the worst-case competitive ratio of $\alg$, we may assume without loss of generality that $\alg$ is deterministic.

Let $p_0$ be the maximum number of processors that $\alg$ ever simultaneously gives work
over all input TAPs. We claim that there is some
input TAP on which $\alg$ does arbitrarily poorly compared to $\opt$ in
terms of mean response time.

Consider a sequence of tasks  $\mathcal{T}$ that causes $\alg$ to choose to
have $p_0$ processors in use at some point in time $t$. Without loss of generality, 
we may assume that all $p_0$ processors that are in use at time $t$ each have at least
$1$ remaining work. Let $h$ be the $\trt$ that $\opt$ would incur on $\mathcal{T}$.
Now, suppose that at time $t$, we have $\ceil{R \cdot h}$ additional tasks arrive, each with $0$ work.

The $\trt$ of $\opt$ on this TAP is $0+h$, whereas the $\trt$ of
 $\alg$ on this TAP is at least $h+1\cdot R\cdot h$, since all $\ceil{R \cdot h}$ of the new tasks
 will have to wait for time at least $1$ before $\alg$ begins them. Hence
 $\alg$'s competitive ratio on this TAP is at least $R$.
\end{proof}

\section{Open Questions}\label{sec:conclusion}

We conclude by discussing open questions and conjectures. 

\paragraph{Questions about awake time} We conjecture that \cref{prop:evilLowerBound} should be tight.

\begin{conjecture}
There exists a deterministic $\phi$-competitive scheduler for awake time.
\end{conjecture}

More concretely, let us propose a scheduler \texttt{GoldenAlg} that we suspect is $\phi$ competitive. 
For simplicity of exposition we assume that $\opt$'s awake time is its completion time on the TAP in question, i.e., $\opt$ does not have ``gaps'' of time when there are no available uncompleted jobs.
It is straightforward to adapt \texttt{GoldenAlg} to handle TAPs with $\opt$ gaps because \texttt{GoldenAlg} will simulate $\opt$ and in particular will know when these gaps are. With this knowledge an appropriately modified version of \texttt{GoldenAlg} would be $\phi$ competitive on arbitrary TAPs assuming that the version of \texttt{GoldenAlg} described below is $\phi$ competitive on TAPs without gaps.

Let $\opt_n$ be an optimal offline schedule for the first $n$ tasks $\mathcal{T}_1^n$. Note that $\opt_{n},\opt_{n+1}$ may make very different decisions, which is part of the challenge for an online scheduler. The \texttt{GoldenAlg} scheduler makes decisions by comparing to $\opt_n$.

    \texttt{GoldenAlg} maintains a \defn{serial pool} 
    of tasks that it has decided to run in serial,
  and a \defn{parallel pool} of tasks that it has tentatively decided
  to run in parallel. \texttt{GoldenAlg} manages the serial pool by running the jobs with the most remaining work first at each time step.
  Crucially, \texttt{GoldenAlg} only runs a single task from
  the parallel pool at a time, so all but one of the
  tasks in the parallel pool have not actually been started.
  Thus, if \texttt{GoldenAlg} desires, it can freely move one of these
  not-yet-started tasks to the serial pool.
  Tasks default to the parallel pool but \texttt{GoldenAlg} greedily moves 
  tasks to the serial pool as soon as it is sure that this will not 
  instantly cause the awake time to be too large.
  Formally, if the awake time incurred so far is $t$, if $n$ tasks have arrived
  so far, and if $\tau_i$ is some not-yet-started task in the parallel pool, then \texttt{GoldenAlg} uses the following logic: If $\sigma_i + t < \phi \cdot \T^n_{\opt_n}$, move
      $\tau_i$ to the serial pool. Finally, if there is no parallel job running but 
  there are tasks present in the parallel pool \texttt{GoldenAlg} schedules the earliest
  arrived task (if any) from the parallel pool in parallel.
  At each time step this parallel job is allocated any processors not allocated to serial jobs.

  \texttt{GoldenAlg} is clearly $\phi$-competitive on any TAP
  that results in \texttt{GoldenAlg} being jagged. However, TAPs that
  result in \texttt{GoldenAlg} being balanced are more difficult. 
  An inductive argument seems challenging because we can 
  re-arrange our work. A more direct combinatorial argument 
  seems more promising but has eluded us.

  We also pose open questions regarding the optimal competitive ratios for different types of awake-time schedulers. We begin by considering decide-on-arrival schedulers:
  
\begin{question}
Does there exist a deterministic decide-on-arrival $2$-competitive scheduler for awake time? 
\end{question}

Such a scheduler would imply that Proposition \ref{prop:2lowerbounddecidearrive} is tight. More broadly, as noted in  \cref{sec:awakelower}, we conjecture that there should be a separation between arbitrary deterministic schedulers and deterministic decide-on-arrival schedulers.

Next, we consider parallel-work-oblivious schedulers:
\begin{question}
What is the optimal awake-time competitive ratio achievable by a deterministic parallel-work-oblivious scheduler?
\label{q:unknown}
\end{question}
We suspect that, for \cref{q:unknown}, the optimal competitive ratio is $4 - o(1)$. 

Finally, we consider randomized schedulers:
\begin{conjecture}
Is the optimal awake-time competitive ratio achievable by a randomized scheduler better than the optimal competitive ratio achievable by a deterministic scheduler?
\end{conjecture}
Note, in particular, that if our lower bounds are tight, then a separation should exist: the optimal competitive ratio for deterministic schedulers would be $\phi \approx 1.62$, and the optimal competitive ratio for randomized ones would be $(3 + \sqrt{3})/4 \approx 1.18$.

\paragraph{Questions about mean response time}
In the context of optimizing $\mrt$, there are even more basic questions that remain open. 

We conjecture that speed augmentation is necessary to achieve a competitive ratio of $O(1)$:
\begin{conjecture}
  $1 + \Omega(1)$ speed augmentation is necessary in an
  $\bigO(1)$-competitive scheduler.
\end{conjecture}

Although we have shown that (deterministic) parallel-work-oblivious schedulers perform poorly on $\mrt$, it remains open whether decide-on-arrival schedulers can do well.

\begin{question}
    Can a decide-on-arrival scheduler, with $O(1)$ speed augmentation, achieve $O(1)$ competitive ratio?
\end{question}

Finally, if algorithms for this problem are to be made practical, then one important direction to focus on is simplicity. This motivates the following question: 
\begin{question}
  Is there a simpler scheduler that is still $\bigO(1)$
  competitive for $\mrt$ with $\bigO(1)$ speed?
\end{question}

An anonymous reviewer suggested one possible direction that would could try, which is to (1) virtually simulate running both jobs at once, and then (2) actually run whichever finishes first (with speed augmentation). The challenge that arises from this approach is that, by the time we finally decide which job to run, we may be significantly past the task's original arrival time. Thus, to make this approach work, one would likely need a much stronger version of Lemma \ref{lem:silly-serious2}, allowing one to analyze settings in which every job has its arrival time delayed based on its simulated completion time. It is conceivable that, in order to prove such a lemma, we could take technical inspiration from work on other scheduling problems (see, e.g., Theorem 3.1 of \cite{bansal2010better} or Theorem 5 of
  \cite{moseley2011scheduling}).

\paragraph{Questions about offline scheduling algorithms}
We conclude with two open problems about \emph{offline} scheduling algorithms.

The first question concerns the scheduling of DTAPS (i.e., TAPS with job-arrival dependencies) to optimize awake time. Our results in \cref{sec:DTAP} establish that the optimal online competitive ratio is $\Theta(\sqrt{p})$ for this problem. However, in some settings, even an \emph{offline} scheduler could be useful, for example, if a compiler understands the parallel structure of the program that it is compiling and needs to decide for each component of the program whether to compile a serial version or a parallel version. 

\begin{question}
Is there a polynomial-time offline algorithm that produces an $O(1)$-competitive DTAP schedule?    
\end{question}

Our final question concerns the offline scheduling of TAPS. We know from  \cref{sec:awakelower} that any online awake-time scheduler must incur a competitive ratio of $1 + \Omega(1)$. But what is the best competitive ratio achievable by a polynomial-time offline scheduler?

\begin{conjecture}
Constructing an optimal (offline) schedule for awake time is NP-hard.
\end{conjecture}




\bibliographystyle{ACM-Reference-Format}
\bibliography{refs}


\end{document}